\documentclass[9pt]{elife}
\newtoggle{prolife}
\toggletrue{prolife}

\usepackage{lipsum} 
\usepackage[version=4]{mhchem}
\usepackage{siunitx}
\DeclareSIUnit\Molar{M}

\usepackage{amsmath}
\usepackage{amsfonts}
\usepackage{amsthm}
\usepackage{amssymb}
\usepackage{graphicx}
\usepackage{subcaption}
\usepackage[colorinlistoftodos]{todonotes}
\usepackage{algorithmic}
\usepackage{algorithm}
\usepackage{mathtools}
\usepackage{bbm}
\usepackage{color}
\graphicspath{ {figures/} }
\usepackage{natbib}
\usepackage{booktabs}
\usepackage{multirow}
\usepackage{pdflscape}

\newcommand{\children}{\mathrm{children}}
\newcommand{\other}{\mathsf{other}}

\newtheorem{lemma}{Lemma}

\newtheorem{assumption}{Assumption}

\newtheorem{definition}{Definition}
\newtheorem{theorem}{Theorem}
\newtheorem{approximation}{Approximation}

\DeclareMathOperator*{\Order}{Order}
\DeclareMathOperator*{\Orderlist}{Orderlist}
\DeclareMathOperator*{\Leaves}{Leaves}
\DeclareMathOperator*{\Desc}{Desc}

\DeclareMathOperator{\N}{\ensuremath{\mathtt{N}}}
\DeclareMathOperator*{\LL}{\ensuremath{\mathtt{L}}}
\DeclareMathOperator{\C}{\ensuremath{\mathtt{C}}}

\DeclareMathOperator*{\argmax}{arg\,max}

\DeclareMathOperator*{\pos}{pos}

\DeclareMathOperator*{\Alleles}{Alleles}

\DeclareMathOperator*{\Apply}{Apply}

\DeclareMathOperator*{\AncState}{AncState}

\DeclareMathOperator*{\TargStat}{TargStat}
\DeclareMathOperator*{\IT}{IT}
\DeclareMathOperator*{\WC}{WC}
\DeclareMathOperator*{\TT}{TT}

\DeclareMathOperator*{\SG}{SG}
\DeclareMathOperator*{\SGWC}{SGWC}
\DeclareMathOperator*{\ow}{otherwise}

\DeclareMathOperator{\Pen}{Pen}
\DeclareMathOperator*{\shared}{shared}

\DeclareMathOperator*{\NBinom}{NB}
\DeclareMathOperator*{\boost}{boost}
\DeclareMathOperator*{\lump}{lump}

\usepackage{ulem}
\title{Estimation of cell lineage trees by maximum-likelihood phylogenetics}

\author[1,2*]{Jean Feng}
\author[2,3]{William S DeWitt III}
\author[3]{Aaron McKenna}
\author[1]{Noah Simon}
\author[1]{Amy Willis}
\author[2,3*]{Frederick A Matsen IV}
\affil[1]{Department of Biostatistics, University of Washington, Seattle, United States}
\affil[2]{Computational Biology Program, Fred Hutchinson Cancer Research Center, Seattle, United States}
\affil[3]{Department of Genome Sciences, University of Washington, Seattle, United States}

\corr{jeanfeng@uw.edu}{JF}
\corr{matsen@fredhutch.org}{FAM}

\presentadd[]{Seattle, Washington}

\begin{document}
\maketitle

\section*{Abstract}

CRISPR technology has enabled large-scale cell lineage tracing for complex multicellular organisms by mutating synthetic genomic barcodes during organismal development.
However, these sophisticated biological tools currently use ad-hoc and outmoded computational methods to reconstruct the cell lineage tree from the mutated barcodes.
Because these methods are agnostic to the biological mechanism, they are unable to take full advantage of the data's structure.
We propose a statistical model for the mutation process and develop a procedure to estimate the tree topology, branch lengths, and mutation parameters by iteratively applying penalized maximum likelihood estimation.
In contrast to existing techniques, our method estimates time along each branch, rather than number of mutation events, thus providing a detailed account of tissue-type differentiation.
Via simulations, we demonstrate that our method is substantially more accurate than existing approaches.
Our reconstructed trees also better recapitulate known aspects of zebrafish development and reproduce similar results across fish replicates.

\section{Introduction}

Recent advancements in genome editing with CRISPR (clustered regularly interspaced short palindromic repeats) have renewed interest in the construction of large-scale cell lineage trees for complex organisms \citep{McKennaaaf7907, Woodworth2017-jc, Spanjaard2018-ge, Schmidt2017-ol}.
These lineage-tracing technologies, such as the GESTALT method \citep{McKennaaaf7907} that we focus on here\footnote{Genome Editing of Synthetic Target Arrays for Lineage Tracing}, inject Cas9 and single-guide RNA (sgRNA) into the embryo of a transgenic organism harboring an array of CRISPR/Cas9 targets separated by short linker sequences (barcodes).
These barcodes accumulate mutations because Cas9 cuts are imperfectly repaired by non-homologous end joining (NHEJ) during development while Cas9 and sgRNA are available.
The resulting mutations are passed from parent cell to daughter cell, which thereby encodes the ontogeny.
Mutated barcodes are later sequenced from the organism, and computational phylogenetic methods are then used to estimate the cell lineage tree.
Because these barcodes have great diversity, GESTALT provides researchers with rich data with the potential to reveal organism and disease development in high resolution.

The most common phylogenetic methods used to analyze GESTALT data are Camin-Sokal (C-S) parsimony \citep{Camin1965-zt} and the neighbor-joining distance-based method \citep{Saitou1987-bw}.
However these methods are blind to the operation of the GESTALT mutation process, so the accuracy of the estimated trees are poor \citep{Salvador-Martinez2018-dw}.
In addition, existing methods supply branch length estimates in terms of an abstract notion of distance rather than time, limiting their interpretability.
Therefore, these estimated trees only provide ordering information between nodes on the same lineage, but not for nodes on parallel lineages.
In addition, C-S parsimony is unable to distinguish between equally parsimonious trees, so obtaining a single tree estimate is difficult in practice: We find over ten thousand parsimony-optimal trees for existing datasets.
To address these challenges, we set out to develop a statistical model of the mutation process, allowing us to estimate branch lengths that correspond to time as well as the mutation parameters.

No appropriate likelihood model is currently available for GESTALT because CRISPR arrays violate many classical statistical phylogenetic assumptions.
First, Cas9 enzymes may cut two targets simultaneously with the entire intervening sequence deleted during NHEJ.
In addition, once the nucleotide sequence for a target is modified, Cas9 is no longer able to cut the target.
Thus sites are not independent, the mutation process is irreversible, and cuts can introduce long insertion and/or deletions.
In contrast, the classical phylogenetic assumptions are that individual nucleotide positions are independent and that the mutation process is reversible and only introduces point mutations \citep{Felsenstein2004-qq,Yang2014-ka}.
Finally, these types of methods assume that there are many independent observations --- their estimates are unstable when the effective sample size is much smaller than the number of parameters to estimate \citep{Goolsby2016-xg, Adams2018-uu, Julien2018-yt}.

In this paper, we introduce GAPML (\underline{G}ESTALT \underline{a}nalysis using \underline{p}enalized \underline{M}aximum \underline{L}ikelihood), a statistical model for GESTALT and tree-estimation method (including topology and branch lengths) by an iterative procedure based on maximum likelihood estimation.
We model barcode mutations as a two-step process: Targets are cut according to a continuous time Markov chain, immediately followed by random insertions or deletions of nucleotides (indels).
Our method does not rely on the aforementioned assumptions.
Instead we introduce the following assumptions tailored to the GESTALT setting:
\begin{itemize}
\item(\emph{outermost-cut}) an indel is introduced by cuts at the outermost cut sites
\item(\emph{target-rate}) the cut rates only depend on which targets are active (i.e.\ able to be cut)
\item(\emph{indel-probability}) the conditional probability that an indel is introduced only depends on which targets were cut.
\end{itemize}
From these assumptions, we show that the Markov process is ``lumpable`` and the aggregated process is compatible with Felsenstein's pruning algorithm, thereby enabling efficient computation of the likelihood \citep{Kemeny1976-ll, Felsenstein1981-zs}.
Since only a small number of barcodes are usually available in practice, we propose a regularization method on the branch length and mutation parameters to stabilize and improve estimates.
Our method extends maximum-likelihood phylogenetic methods with branch-length penalties \citep{Kim2008-rk} to the setting where the tree topology is unknown.

We validate our method on simulated and empirical data.
In simulations, our method is more accurate than current tree-estimation methods.
In addition, we reconstruct cell lineage trees of transgenic zebrafish from \citet{McKennaaaf7907} and show that our trees better reflect the known biology of zebrafish development.
Based on these results, we conclude that with appropriate statistical techniques it is possible to reconstruct an accurate cell lineage tree with current GESTALT technology, which addresses some concerns raised in \citet{Salvador-Martinez2018-dw}.
Our simulation engine and estimation method are available on Github (\url{https://github.com/matsengrp/gestaltamania}).

\section{Results}

\subsection{Brief description of our probabilistic GESTALT evolution model}
\begin{figure}
	\centering
	\includegraphics[width=0.7\linewidth]{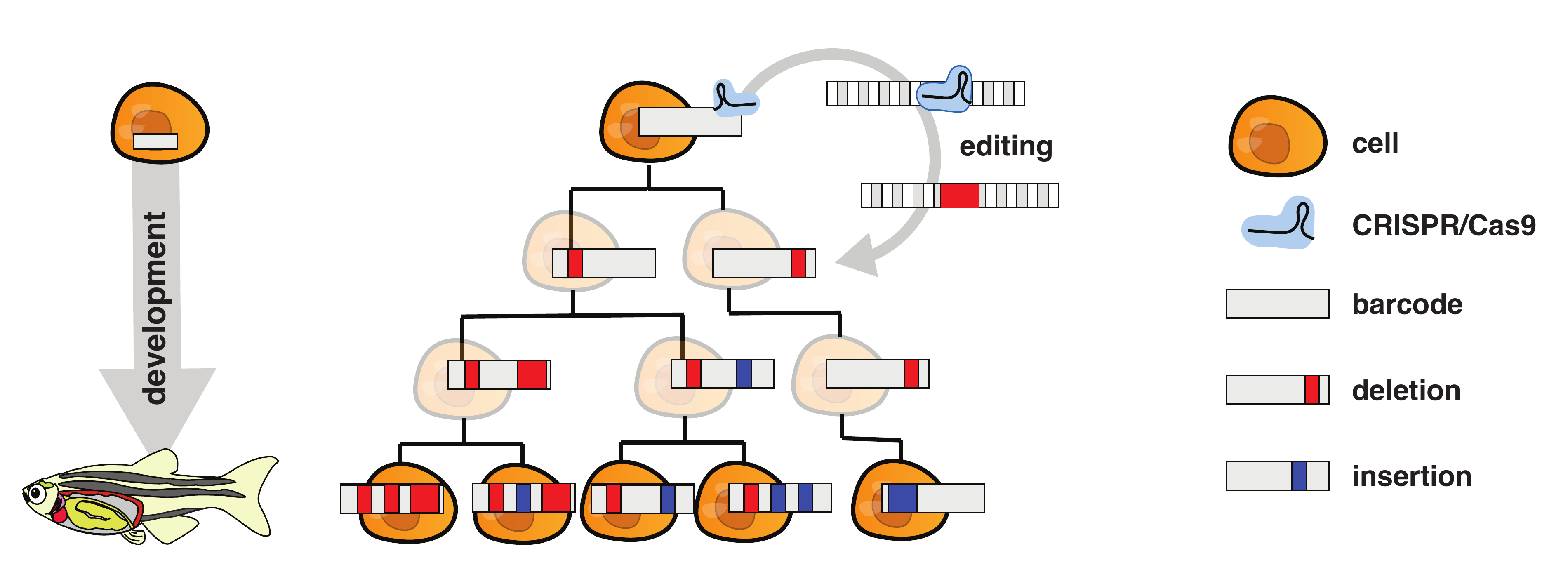}
	\caption{
		An unmodified array of CRISPR/Cas9 target sites (i.e., a GESTALT barcode) is engineered into an organism's genome. CRISPR/Cas9 enzyme complex with corresponding guide sequences are directed to make double-stranded breaks in the barcode. These breaks are repaired in an error-prone fashion resulting in insertions and deletions at target sites. These insertions and deletions will accumulate in a lineage specific fashion, passed from mother to daughter cell, and further insertions and deletions can add additional information. These integrated barcodes can then be recovered by DNA sequencing at the timepoint of interest.
	}
	\label{fig:gestaltexplain}
\end{figure}

We model the GESTALT barcode (see Figure~\ref{fig:gestaltexplain}) as a continuous time Markov chain where the state space is the set of all nucleotide sequences.
A state transition is an instantaneous event where either (1) an unmodified target is cut then the repair process inserts/deletes nucleotides around the cut site, or (2) two unmodified targets are cut, the intervening sequence is removed, and the repair process inserts/deletes nucleotides around the cut sites.
The transition rate between barcode sequences depends on the entire sequence and each target is associated with a separate cut rate.
If multiple copies of the barcode are used, we assume the barcodes are on separate chromosomes or are sufficiently far apart that they act in an independent and identically distributed (iid) manner.

We use this Markov model for GESTALT barcodes evolving along a cell lineage tree where the vertices represent cell divisions and the edge lengths represent time between cell divisions.
The full cell lineage tree describes the relationships of all cells in the organism.
Since we only collect a small sample of all the cells, our goal is to recover the subtree describing the development of the observed sequences.

To estimate this subtree, our method needs to calculate the likelihood of possible trees and model parameters, which requires an enumeration of the possible barcodes at each internal node.
However a full enumeration is infeasible.
For example, a double cut (transition (2) described above) could remove one or more targets, which could have themselves been modified in an infinite number of possible ways before the double cut erased this history.

We have carefully chosen our assumptions to make this likelihood tractable yet maintain biological realism.
Briefly, under the \emph{target-rate} and \emph{indel-probability} assumptions, we can group states together if they share the same set of unmodified targets to calculate the likelihood more efficiently, a property known more formally known as ``lumpability'' (Figure~\ref{fig:lumpability}).
Since the number of targets in a barcode is typically small (e.g. 10 targets per barcode in \citet{McKennaaaf7907, Schmidt2017-ol}), calculating the likelihood becomes computationally feasible.
In addition, the \emph{outermost-cut} assumption allows us to exclude many groups from the likelihood computation so that the number of enumerated groups at most internal nodes is typically linear in the number of unique indels.
\begin{figure}
	\begin{center}
	\includegraphics[width=0.5\linewidth]{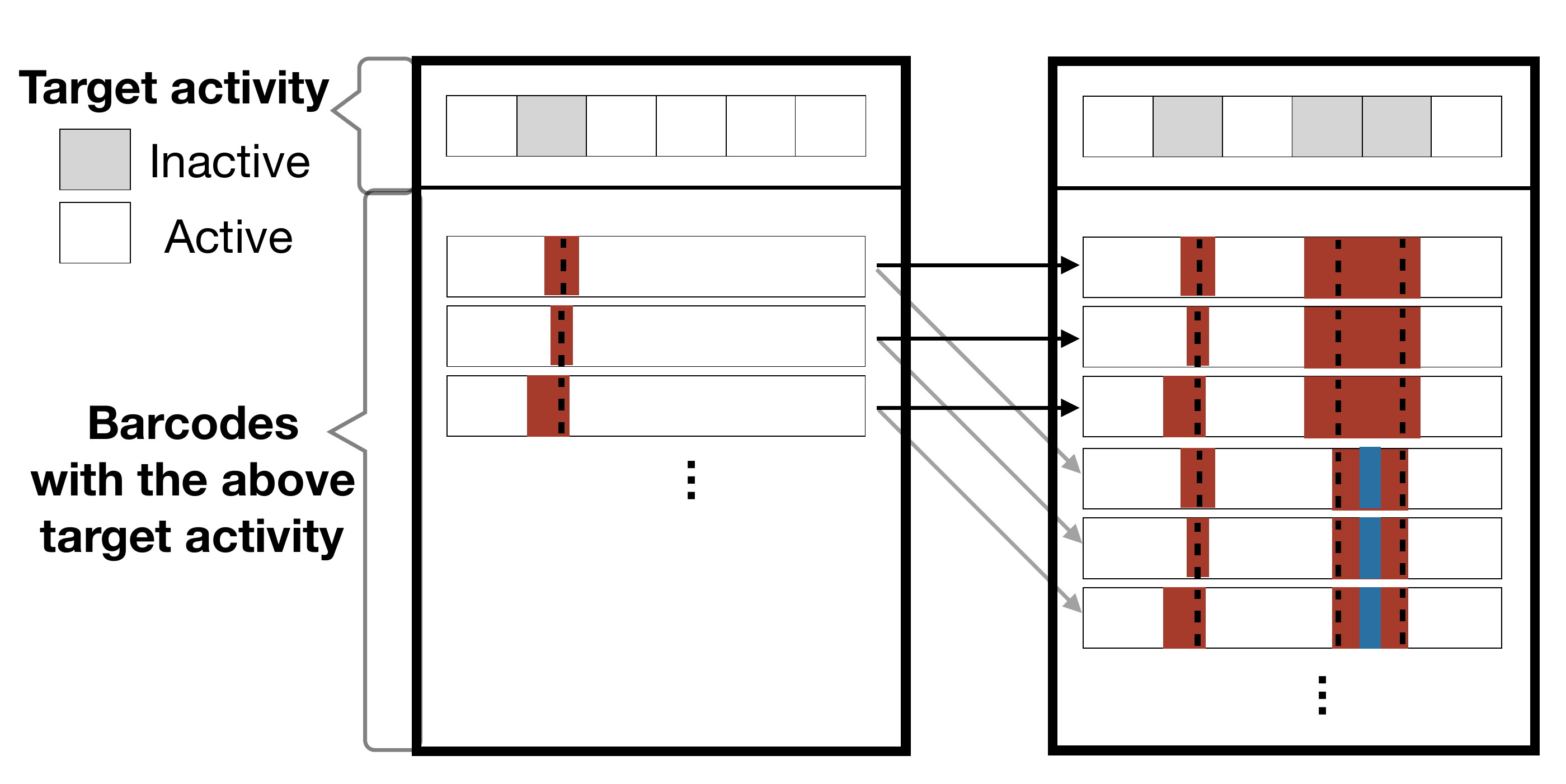}
	\end{center}
	\caption{
		An example of lumping together barcodes that share the same target activity.
		The two outer boxes correspond to two of the lumped states.
		The left box is the grouped state for possible ancestral barcode states where the second target is no longer active, while the right box represents when the second, fourth, and fifth targets are no longer active.
		The arrows represent possible transitions and the color represents the transition rates.
		Notice that each barcode in the left box has the same set of outgoing arrows.
		To show that the states are lumpable, we show that the total transition rate out of a barcode in the left box to the right box is the same for all barcodes in the left box.
	}
	\label{fig:lumpability}
\end{figure}

\subsection{A maximum-likelihood tree estimation procedure}
\label{sec:mle_overview}
We follow current best practice for maximum-likelihood phylogenetics by optimizing the tree and mutation parameters of our model using a hill-climbing iterative search over tree space.
First, we initialize the tree topology by selecting a random parsimony-optimal tree from C-S parsimony.
At each subtree prune and regraft (SPR) iteration, we select a random subtree and regraft where the penalized log likelihood is highest (Figure~\ref{fig:algo_spr}).
The method stops when the tree no longer changes.
At each iteration, we only consider SPR moves that preserve the parsimony score as we have found that the parsimony-optimal trees tend to have the highest likelihoods (Figure~\ref{fig:parsimony_log_lik}).
The entire algorithm is presented in Algorithm~\ref{algo:whole_thing}.
We discuss some important details of our method below.

We maximize a penalized log likelihood as opposed to the unpenalized version since the latter tends to give unstable and inaccurate estimates when the dataset is generated by a small number of barcodes.
In particular, the length of the leaf branches and the variance of the target rates tend to be overestimated in such settings.
Thus we use a penalty function that discourages large differences in branch lengths and target cut rates.
Penalization introduces a slight complication since certain candidate SPR moves have naturally larger penalties.
In order to make the penalty comparable between candidate SPR moves, we randomly select a leaf in the subtree and apply the candidate SPR moves only to that single leaf.
When scoring the SPR moves, the penalty is calculated for the shared subtree, i.e. the tree where we ignore the random leaf.
Finally, we regraft the entire subtree where the penalized log likelihood is highest.

Our method is able to estimate the tree at a finer resolution than existing methods (Figure~\ref{fig:tree_resolution}).
The most commonly used method, C-S parsimony, produces estimates at the coarsest resolution: For nodes where the ordering is ambiguous, the method simply groups them under a single parent node.
This commonly results in tree estimates with many multifurcating nodes (nodes with 3+ children) that have ten or more children.
Our method uses the estimated model parameters to estimate the order and time of ambiguous nodes by projecting the subtree onto the space of caterpillar trees (Figure~\ref{fig:caterpillar}).
By producing tree estimates at a finer resolution, our method allows researchers to learn more about the structure of the true cell lineage tree.
In addition, taking advantage of the irreversibility property, we efficiently estimate the branch ordering within the caterpillar trees by solving a single optimization problem, rather than considering each possible ordering separately.

\begin{figure}
	\begin{subfigure}{\textwidth}
		\begin{center}
			\includegraphics[width=0.6\textwidth]{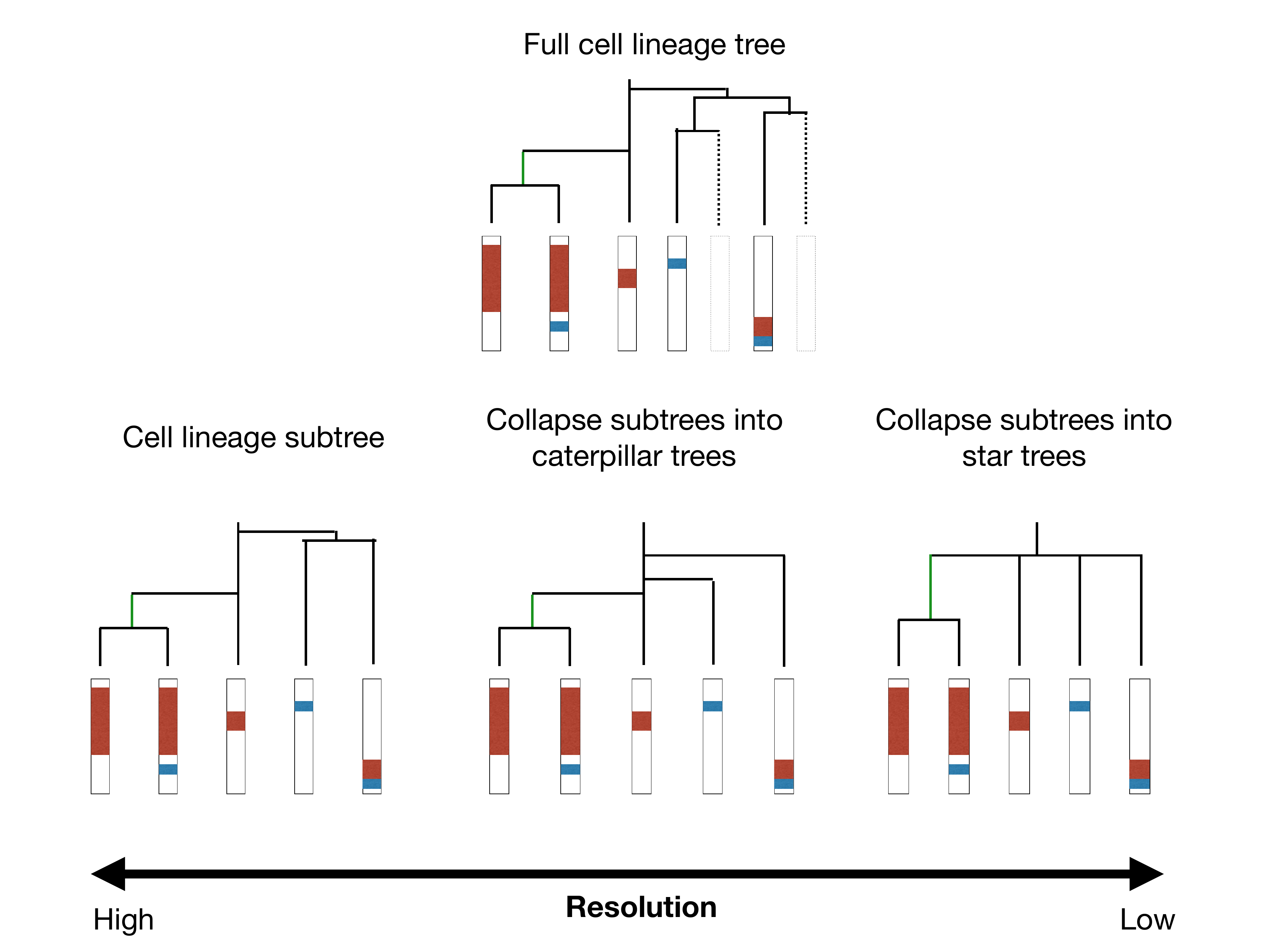}
		\end{center}
		\caption{
	            We show the subtree of a full cell lineage tree (top) at different resolutions.
		    The highest resolution preserves the bifurcating tree structure (left).
		    The lowest resolution preserves very coarse order information by collapsing a subtree into a multifurcating node (right).
		    In between these two resolutions, we can project the tree onto the space of caterpillar trees and preserve the ordering information between nodes (middle).
		}
		\label{fig:tree_resolution}
	\end{subfigure}
	\begin{subfigure}{0.45\textwidth}
		\begin{center}
			\includegraphics[width=\textwidth]{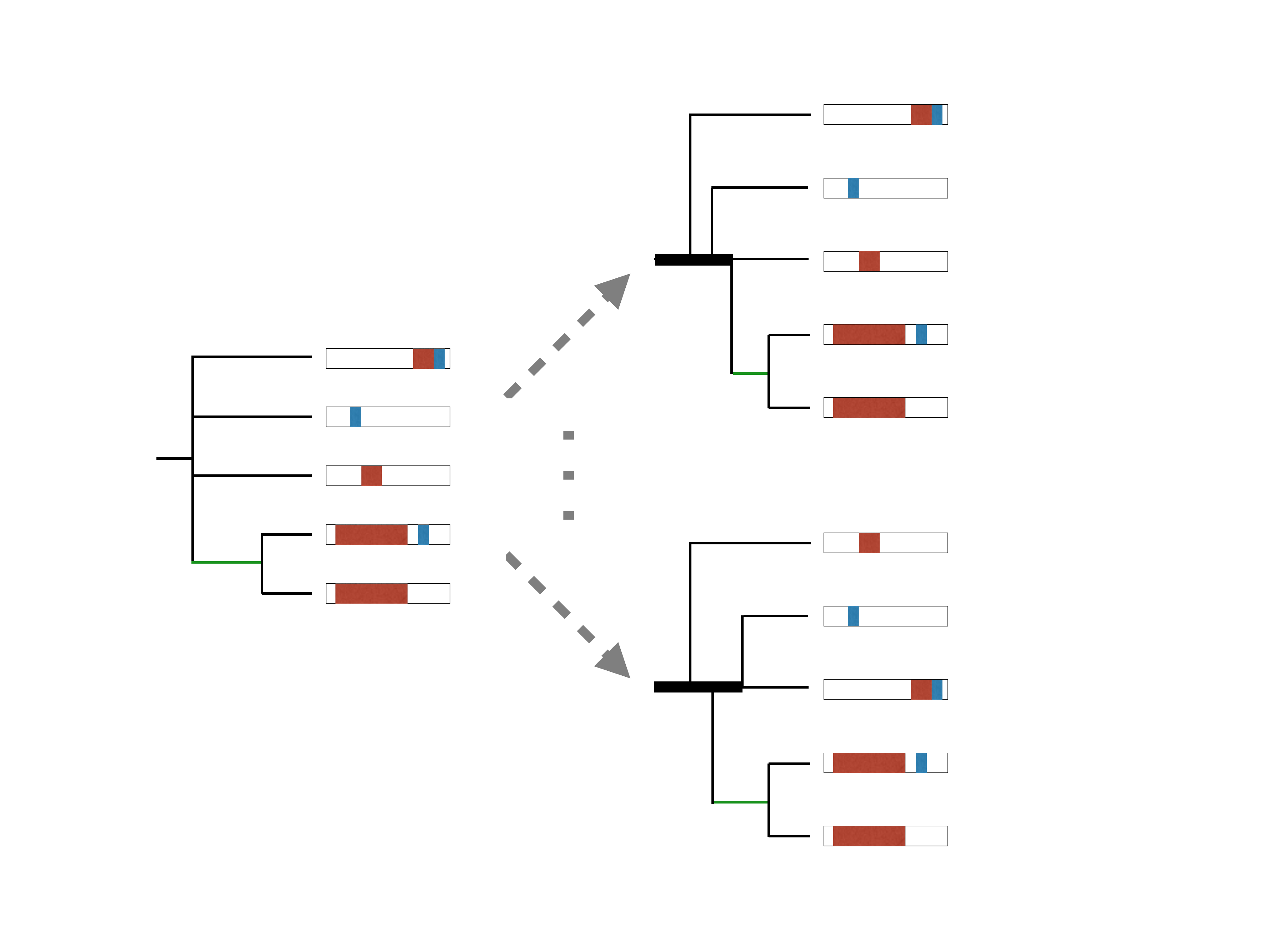}
		\end{center}
		\caption{
			We resolve each multifurcation in the tree into a caterpillar tree, which places all the children nodes along a central path.
			This central path, called a caterpillar spine, is indicated by the bold lines in the trees on the right.
			There are many possible orderings in a caterpillar tree.
			Here we show two such orderings.
			Our method chooses the ordering that maximizes the penalized log likelihood.
		}
		\label{fig:caterpillar}
	\end{subfigure}
	\hspace{0.1in}
	\begin{subfigure}{0.45\textwidth}
		\begin{center}
			\includegraphics[width=\textwidth]{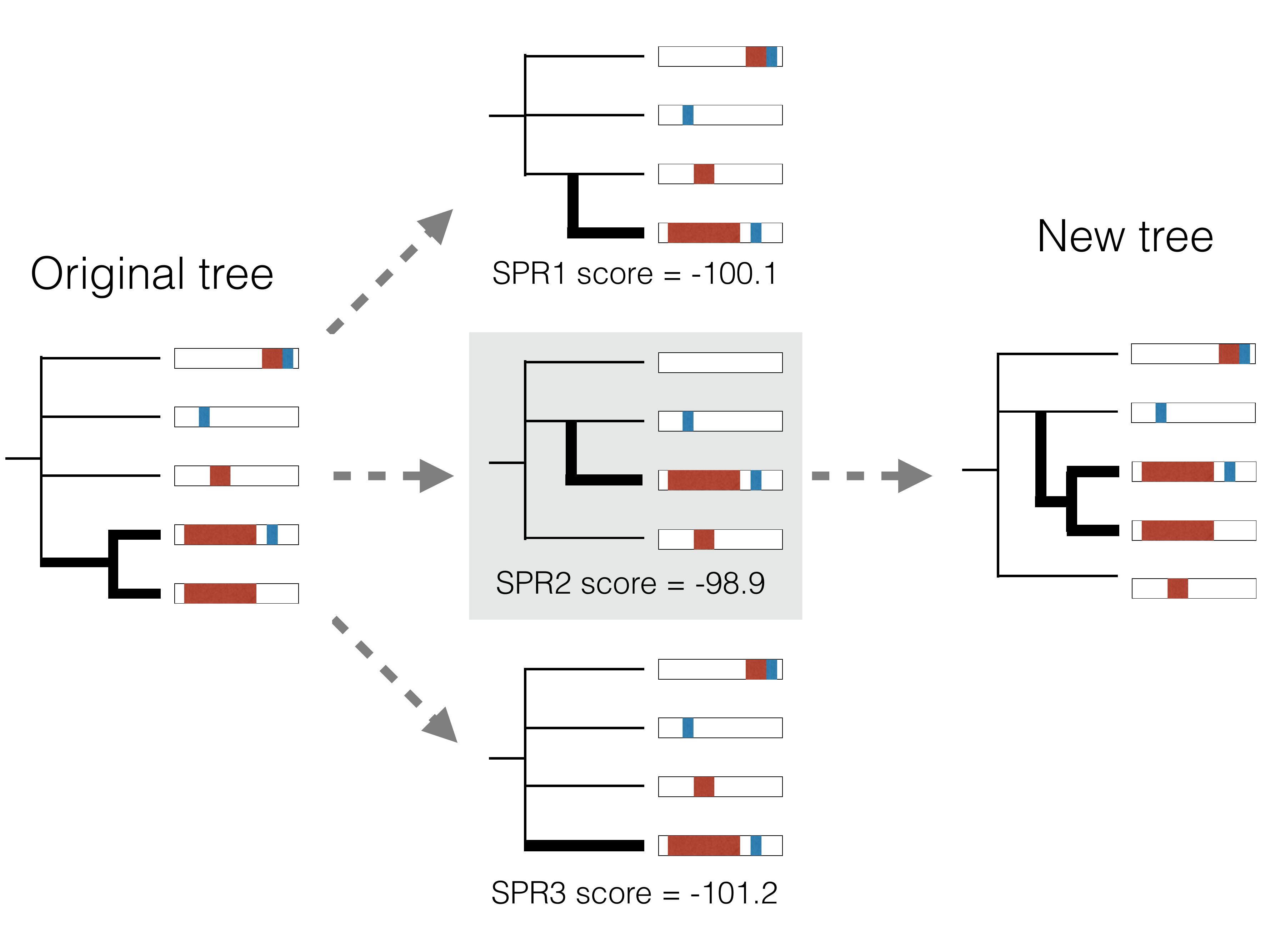}
		\end{center}
		\caption{To tune the tree topology, we select a random subtree (left) and score possible SPR moves that preserve the parsimony score by selecting a random subleaf and calculating the maximized penalized log likelihood of the resulting tree (middle).
			We then update the tree by applying the SPR move with the highest score (right).}
		\label{fig:algo_spr}
	\end{subfigure}
\caption{Overview of our tree estimation method.}
\end{figure}

\begin{algorithm}
	\caption{Cell lineage tree reconstruction for penalty parameter $\boldsymbol{\kappa}$}
	\label{algo:whole_thing}
	\begin{algorithmic}
		\STATE{Initialize tree $\mathbb{T}$. Let the sequenced GESTALT barcodes be denoted $D$.}
		\FOR{Iteration $k$}
		\STATE{Pick a random subtree from $\mathbb{T}$. Select one of the leaves $\C$ of the subtree.}
		\FOR{each possible SPR move involving the subtree that doesn't change the parsimony score (including the no-op)}
		\STATE{Construct $\mathbb{T}'$ by applying the SPR to leaf $\C$; let $\mathbb{T}'_{\shared}$ be the subtree of $\mathbb{T}'$ when excluding $\C$}
		\STATE{
			Set the score of the SPR move as the penalized log likelihood maximized with respect to the branch length parameters $\ell$ and deactivation and indel process parameters $\theta$ and $\beta$, respectively:
			\begin{align*}
			\max_{\ell, \theta, \beta}
			&\log \underbrace{\Pr(D, \text{the barcode is constant along all caterpillar spines}; \mathbb{T}', \ell, \theta, \beta)}_{\text{Approximation to the likelihood}}\\
			& + \underbrace{\Pen_{\boldsymbol{\kappa}}(\mathbb{T}'_{\shared}, \ell, \theta, \beta)}_{\text{Penalty on branch lengths and mutation parameters}}
			\end{align*}
		}
		\ENDFOR
		\STATE{Update the tree $\mathbb{T}$ by performing the SPR move on the subtree that maximizes the score}
		\ENDFOR
	\end{algorithmic}
\end{algorithm}

\subsection{Simulation engine and results}
We built a simulation engine of the GESTALT mutation process during embryonic development.
Since cell divisions during embryonic development begin in a fast metasynchronous fashion and gradually become more asynchronous \citep{Moody1998-uq}, the simulation engine generates a cell lineage tree by performing a sequence of synchronous cell divisions followed by a birth-death process where the birth rate decays with time.
We mutate the barcode along this cell lineage tree according to our model of the GESTALT mutation process.
The simulation engine can generate data that closely resembles the data collected from zebrafish embryos in  \citet{McKennaaaf7907} (Figure~\ref{fig:simulator}).
We can input different barcode designs into the simulation engine to understand how they affect our ability to reconstruct the cell lineage tree.

We used our simulation engine to assess the validity and accuracy of the estimated model parameters and tree.
Because our method infers branch lengths, we evaluate the accuracy using two metrics that include branch length information: BHV distance \citep{Billera2001-ii} and internal node height correlation (see Figure~\ref{fig:internal_node_height}).
We compare our method to a simpler model-free approach: estimating the tree topology using C-S parsimony \citep{Camin1965-zt} or neighbor-joining (NJ) \citep{Saitou1987-bw} and then applying semiparametric rate smoothing (\texttt{chronos} in the \texttt{R} package \texttt{ape}) to estimate branch lengths \citep{Sanderson2002-gs}.
We will refer to these two approaches as ``CS+chronos'' and ``NJ+chronos.''
We do not compare against the original tree estimates from C-S parsimony and neighbor-joining since those branch lengths correspond to edit distance and have very poor performance according to our two metrics.
Our method consistently outperforms these alternative methods (Figure~\ref{table:big_sim}).
We note that previous \emph{in silico} analyses of GESTALT measure accuracy in terms of the Robinson-Foulds (R-F) distance, which only depends on the tree topology \citep{Salvador-Martinez2018-dw}.
However the R-F distance does not recognize that different tree topologies can be very similar depending on their branch lengths, and is therefore too coarse as a performance metric.

We find, based on the simulations, that our likelihood-based method improves in performance as the number of independent barcodes increases (Figures \ref{fig:consistency}).
In a simulation with a six-target barcode, the estimated tree from a single barcode has internal node height correlation of 0.5 with the true tree whereas using four barcodes increases the correlation to 0.9.
Even though other analyses of GESTALT have recommended increasing the number of targets in a single barcode to improve tree estimation \citep{Salvador-Martinez2018-dw}, it is more effective to increase the number of targets by introducing independent (and identical) barcodes (Figure~\ref{fig:many_vs_one}).

\begin{figure}
	\begin{subfigure}{0.46\textwidth}
		\begin{center}
            \includegraphics[height=0.22\textheight]{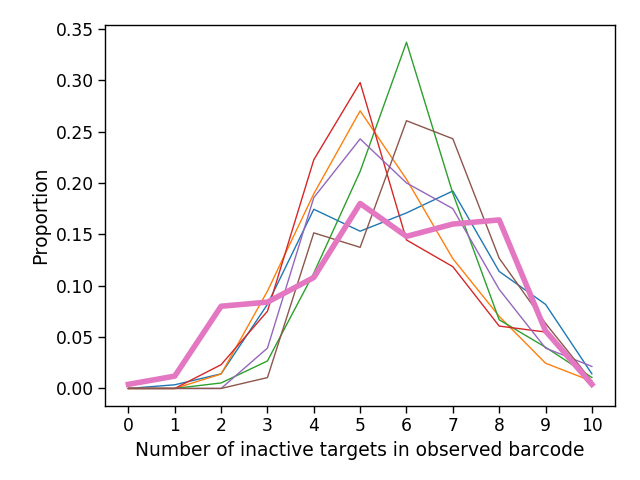}
			\includegraphics[height=0.22\textheight]{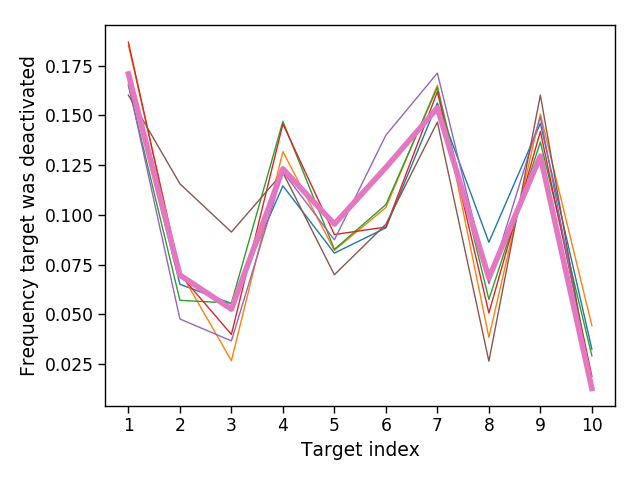}
			\includegraphics[height=0.22\textheight]{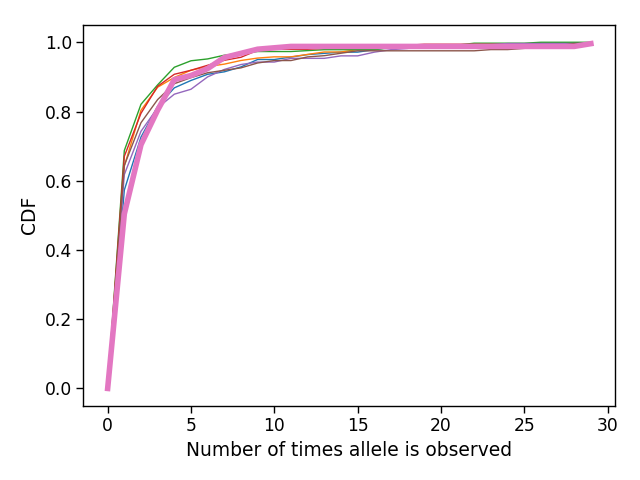}
		\end{center}
		\caption{
			A comparison of summary statistics on the simulated data (each thin line is a replicate; color used only to aid in distinguishing between replicates) vs. 250 randomly selected alleles from the first dome fish (bolded line).
			We generated data from our simulation engine and randomly sampled leaves to obtain around 250 unique alleles.
			The distribution of inactive targets and allele abundances (the number of times an allele is observed) are similar.
		}
		\label{fig:simulator}
	\end{subfigure}
	\hspace{0.4in}
	\begin{subfigure}{0.45\textwidth}
		\begin{center}
			\includegraphics[height=0.6\textheight]{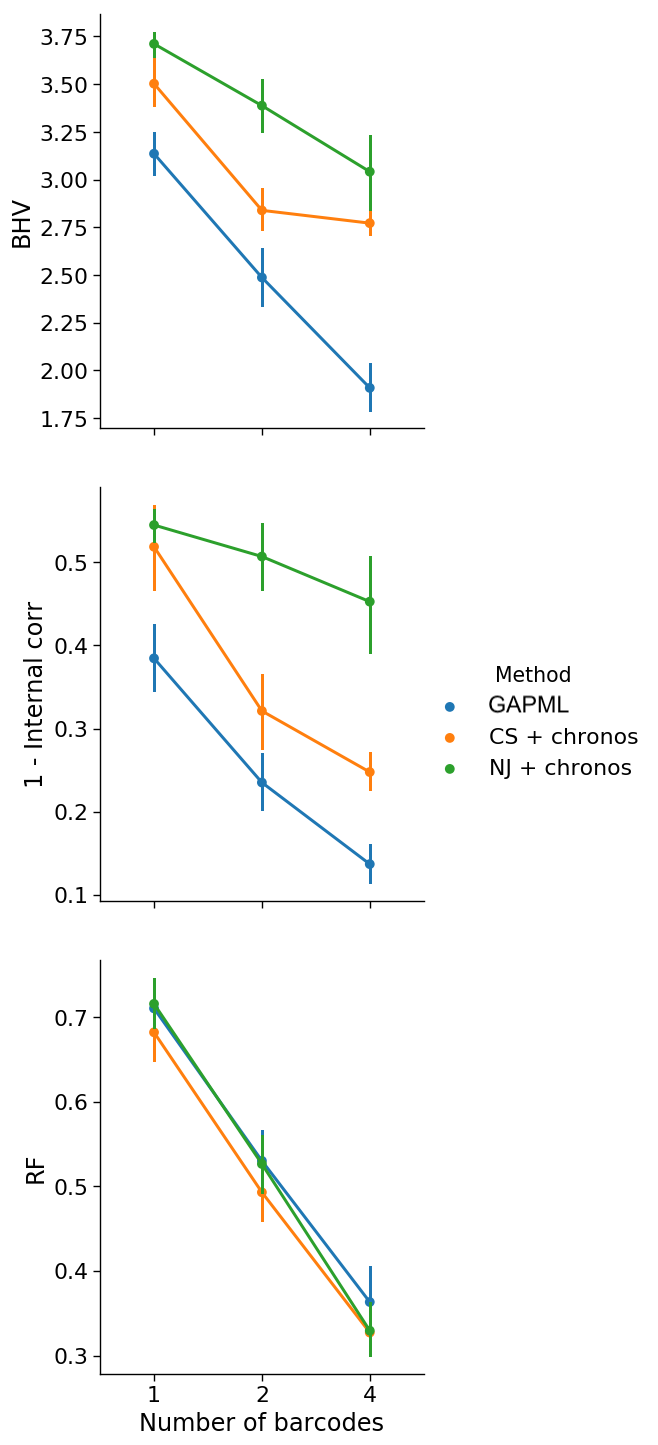}
		\end{center}
		\vspace{-0.1in}
		\caption{
			Results for data simulated from a barcode with six targets and randomly sampled to obtain roughly 100 unique alleles.
			The performance of GAPML improves with the number of barcodes.
			GAPML performs significantly better than the other methods in terms of BHV (top) and the internal node height correlation metrics (middle).
			The methods are hard to distinguish with respect to the Robinson-Foulds (RF) metric (bottom).
		}
		\label{fig:consistency}
	\end{subfigure}
	\begin{subfigure}{\textwidth}
		\begin{center}
        \begin{tabular}{c|cc}
        	Method & BHV & 1 - Internal node correlation\\
        	\toprule
        	GAPML & 5.68 (5.51, 5.85) & 0.45 (0.42, 0.48) \\
        	CS + chronos & 6.39 (6.20, 6.58) & 0.58 (0.52, 0.64) \\
        	NJ + chronos & 8.48 (8.38, 8.58) & 0.66 (0.64, 0.68) \\
        \end{tabular}
	\end{center}
        \caption{Comparison of methods on simulated data using a single barcode with ten targets and around 200 leaves. The 95\% confidence intervals are given in parentheses.}
        \label{table:big_sim}
	\end{subfigure}
	\caption{
       Simulation results. We denote Camin-Sokal parsimony and neighbor-joining with nonparametric rate smoothing as CS+chronos and NJ+chronos, respectively.
	}
	\label{fig:simulation_things}
\end{figure}

\subsection{Improved zebrafish lineage reconstruction}

To validate our method, we reconstructed cell lineages using our method and other tree-building methods on GESTALT data from zebrafish \citep{McKennaaaf7907}.
As the true cell lineage tree is not known for zebrafish, we employed more indirect measures of validity.
For each method, we asked (1) if similar conclusions could be made across different biological replicates and (2) if the tree estimates aligned with the known biology of zebrafish development.
The dataset includes two adult zebrafish where cells were sampled from dissected organs.
The organs were chosen to represent all germ layers: the brain and both eyes (ectodermal), the intestinal bulb and posterior intestine (endodermal), the heart and blood (mesodermal), and the gills (neural crest, with contributions from other germ layers).
The heart was further divided into four samples--- a piece of heart tissue, dissociated unsorted cells (DHCs), FACS- sorted GFP+ cardiomyocytes, and non-cardiomyocyte heart cells (NCs).
In addition, datasets were collected from embryos before gastrulation (dome stage, 4.3 hours post-fertilization (hpf)), at pharyngula stage (30 hpf), and from early larvae (72 hpf),  where the cell type assignments are unknown.

\paragraph{GAPML captures more similar developmental relationships between tissue types across the two adult fish replicates.}
For each estimated tree, we calculated the \emph{distance between tissues} --- the average tree distance between a leaf of one tissue to the closest internal node leading to a leaf from the other tissue, weighted by the allele abundance (Figure~\ref{fig:distance_matrices}).
(All alleles that were found in the blood were removed since blood is found in all dissected organs and can confound the relationship between organs \cite{McKennaaaf7907}.)
Recall that all of the fitting procedures are completely agnostic to any tissue source or cell abundance information.
For a good method, we expect the correlation between tissue distances from the two fish to be close to one.
We tested if the correlations were significant by permuting the cell types and abundances in the estimated trees.
The correlation was 0.770 ($p < 0.001$) using our method whereas `CS+chronos' and `NJ+chronos' had correlations of 0.306 ($p = 0.21$) and -0.325 ($p = 0.22$), respectively.
One might be concerned that our method is consistent across fish replicates because it returns very similar trees regardless of the data.
However, this is not the case: When we re-run our method with randomly permuted cell types and abundances, the average correlation between the tissue distances drops to zero.

\begin{figure}
	\centering
    \includegraphics[width=0.95\textwidth]{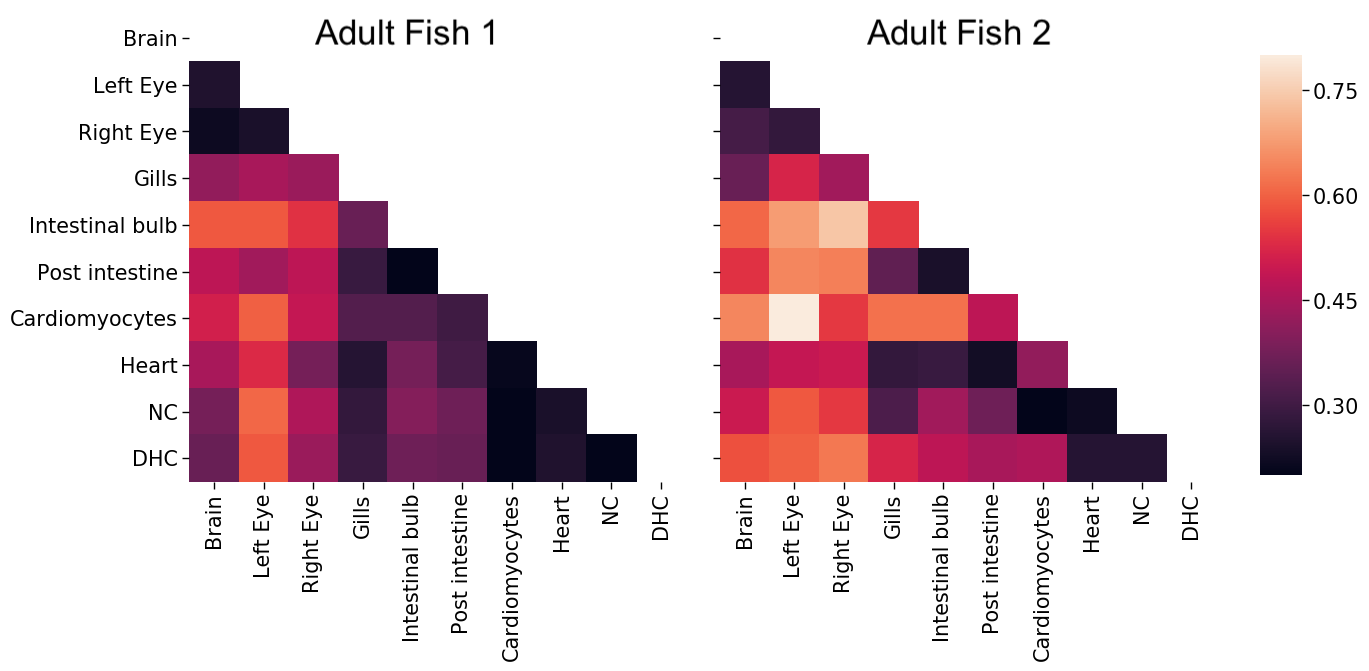}
	\caption{
		The average distance between tissue sources in the estimated trees for adult fish 1 (left) and 2 (right).
		The distance between tissues is the average time from a leaf of one tissue to the closest internal node with a descendant of the other tissue.
		The shading reflects distance, where bright means far and dark means close.
		The tissue distances share similar trends between the two fish.
		For example, the top (brain and eyes) and lower right (heart-related organs) tend to be the darker regions in both distance matrices.
	}
	\label{fig:distance_matrices}
\end{figure}

\begin{figure}
\begin{subfigure}{\textwidth}
\begin{center}
\begin{tabular}{c|c|c|c|c}
		Fish age & $n$ & Barcode version & GAPML Correlation & Empirical average correlation \\
		\toprule
		4 months & 2 & 7 & 0.891 & 0.685 \\
		3 days & 5 & 7 & 0.881 (0.839, 0.982) & 0.688 (0.610, 0.923) \\
		30 hpf & 4 & 6 & 0.309 (0.309, 0.794) & 0.052 (0.052, 0.727) \\
		4.3 hpf & 4 & 7 & 0.931 (0.931, 0.982) & 0.743 (0.717, 0.976) \\
\end{tabular}
\end{center}
\caption{
	Mean Spearman correlation between target lambda rates across fish replicates. 95\% confidence intervals (via bootstrap) shown in parentheses.
}
\label{table:target_lam_correlation}
\end{subfigure}
\begin{subfigure}{\textwidth}
\begin{center}
\includegraphics[width=0.32\linewidth]{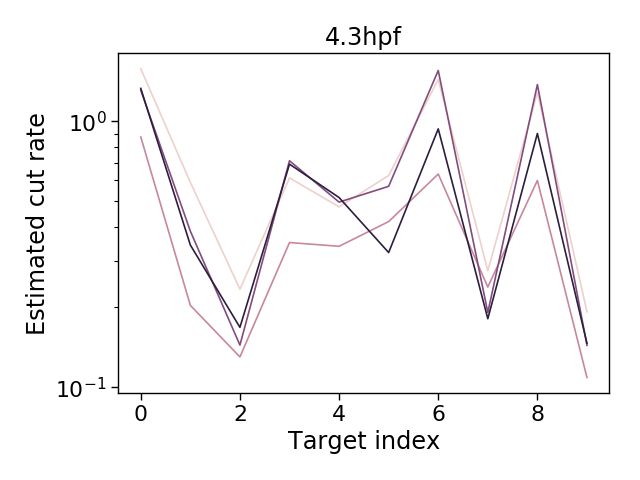}
\includegraphics[width=0.32\linewidth]{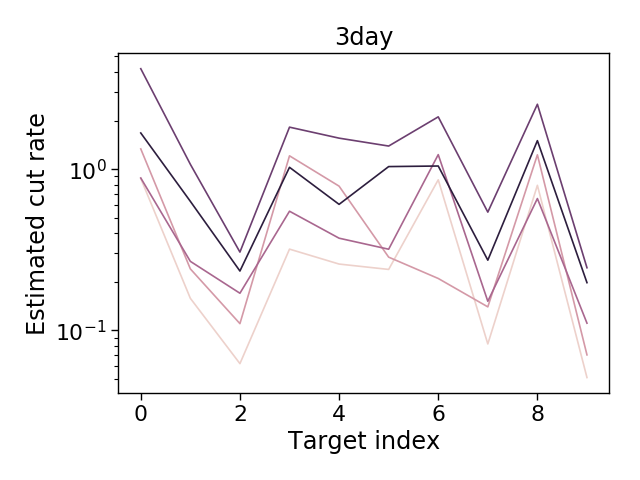}
\includegraphics[width=0.32\linewidth]{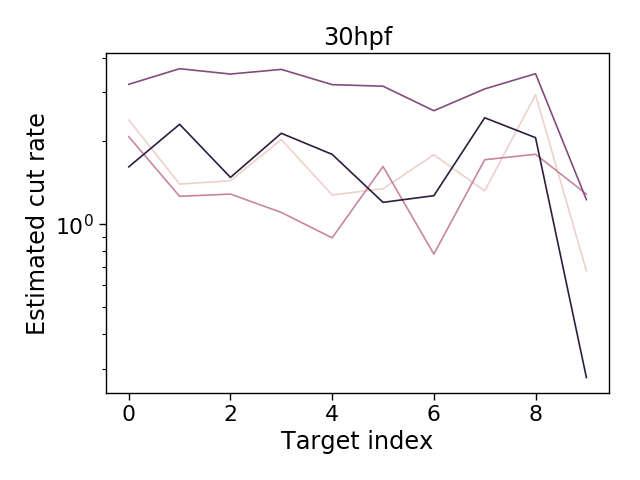}
\end{center}
\caption{
Fitted target lambda rates for fish sampled at 4.3hpf (left), 3 days (middle), and 30 hpf (right), where each colored line corresponds to the estimates for a single fish.
The fish sampled at 4.3hpf and 3 days had the same barcode and share similar target rates.
The 30hpf fish used a different barcode and have different estimated target rates.
}
\label{fig:target_rate}
\end{subfigure}
\caption{Target cut rate estimates are consistent across fish replicates.}
\end{figure}

\subsubsection{GAPML estimates similar mutation parameters across fish replicates.}

For each time point, the fish replicates were traced using the same GESTALT barcode and processed using the same experimental protocol (Table~\ref{table:target_lam_correlation}).
We compared the estimated target rates from our method to those estimated using a model-free empirical average approach where the estimated target cut rate is the proportion of times a cut was observed in that target in the set of unique observed indels.
The average correlation between the estimated target rates from our method were much higher than that for the alternate approach (Figure~\ref{table:target_lam_correlation}).
In fact, we can also compare target cut rates between fish of different ages that share the same barcode, even if the experimental protocols are slightly different.
The 4.3hpf and 3day fish share the same barcode version and we find that the target rate estimates are indeed similar (Figure~\ref{fig:target_rate}).
Again, a possible concern is that our method may have high correlation because it outputs very similar values regardless of the data.
However, the estimated target cut rates were different for fish with different barcode versions.
More specifically, the 30hpf fish used version 6 of the GESTALT barcode whereas the other fish used version 7.
Visually, the target rates look quite different between those in the 30hpf fish and the other fish with the version 7 barcode (Figure~\ref{fig:target_rate}).
Calculating the pairwise correlations between the estimated rates in 30hpf versus 3day fish, the average correlation, 0.416, is quite low and the bootstrap 95\% confidence interval, (0.046, 0.655), is very wide and nearly covers zero.

\paragraph{GAPML recovers both cell-type and germ-layer restriction.}
It is well known that cells are pluripotent initially and specialize during development.
To evaluate recovery of specialization by tissue type, we calculated the correlation between the estimated time of internal tree nodes and the number of descendant tissue types; to evaluate recovery of specialization by germ layer, we calculated the correlation between the estimated time of internal nodes and the number of germ layers represented at the leaves.
(As before, all the estimation methods do not use the tissue origin and germ layer labels.)
Since any tree should generally show a trend where parent nodes tend to have more descendant cell types than their children, we compared our tree estimate to the same tree but with random branch length assignments and randomly permuted tissue types.
Our method estimated much higher correlations compared to these random trees (Table~\ref{table:specialization}).
We show an example of the node times versus the number of descendant cell types and germ layers in Figure~\ref{fig:valid_cell}.
The estimated correlations from the other methods tended to be closer to zero compared to those in GAPML in all cases, except when using `NJ + chronos` to analyze the second adult fish.
However upon inspection, the correlation is high for `NJ + chronos` because it estimates that cells are pluripotent for over 90\% of the fish's life cycle and specialize during a small time slice at the very end.

\begin{figure}
	\begin{subfigure}{\textwidth}
		\begin{center}
		\begin{tabular}{c|c|ccc|ccc}
			Adult & Estimation  & \multicolumn{3}{c|}{\#  tissue types vs time} & \multicolumn{3}{c}{\# germ layers vs time} \\
			Fish & Method & Corr & Random corr & p-value & Corr  & Random corr & p-value \\
			\toprule
			\multirow{3}{0.01in} 1 & GAPML & -0.492 & -0.168 & $<$ 0.001 & -0.421 & -0.124 & $<$ 0.001 \\
			 & CS+chronos & -0.182 & 0.037 & 0.002 & -0.142 & 0.032 & 0.044 \\
			 & NJ+chronos & -0.271 & -0.126 & 0.003 & -0.179 & -0.094 & 0.084 \\
			\midrule
			\multirow{3}{0.01in} 2  & GAPML  & -0.493 & -0.220 & $<$0.001 & -0.410 & -0.151 & 0.002 \\
			& CS+chronos & -0.389 & 0.070 & 0.001  & -0.397 & 0.090 & $<$ 0.001 \\
			& NJ+chronos  & -0.621 & -0.236 & $<$0.001 & -0.475 & -0.183 & 0.001 \\
		\end{tabular}
		\end{center}
		\caption{
			Estimated correlations between the number of descendant cell types/germ layers vs. the time of internal nodes in the tree.
			Since some tree topologies naturally have higher correlations, we also show the correlation when cell types are shuffled and branch lengths are randomly assigned.
			The p-value for each tree is calculated with respect to their respective randomly shuffled trees.
		}
		\label{table:specialization}
	\end{subfigure}
	\begin{subfigure}{\textwidth}
		\includegraphics[width=0.49\textwidth]{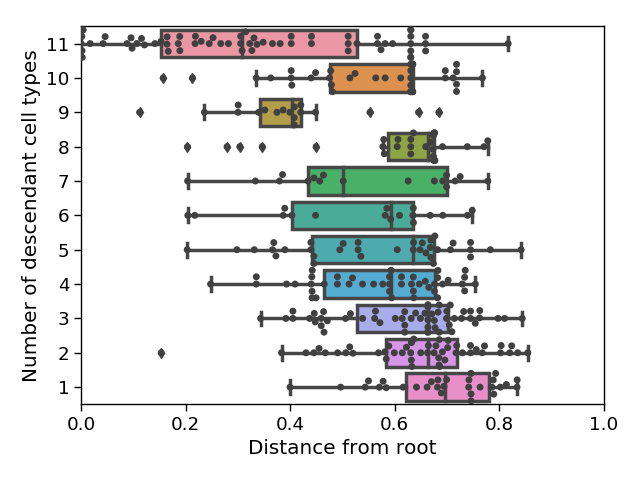}
		\includegraphics[width=0.49\textwidth]{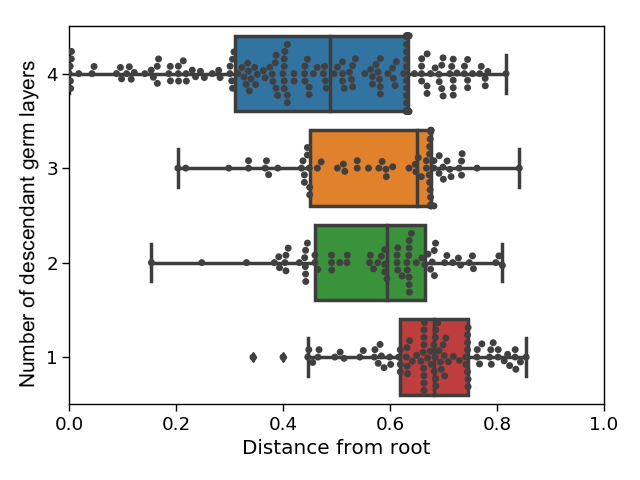}
		\caption{
			Box plots of the internal node times in the estimated tree for the first adult fish using GAPML, where nodes are grouped by the number of descendant cell types (left column) and the number of descendant germ layers (right column).
		}
		\label{fig:valid_cell}
	\end{subfigure}
	\caption{
		Estimated relationships between node times and number of descendant cell types and germ layers in the two adult fish for the different methods.
	}
	\label{fig:intern_cell}
\end{figure}

\subsection{Analysis of the zebrafish GESTALT data}

In this section, we analyze the fitted trees of the adult zebrafish in more detail.
Our primary goals are to (1) check if summaries concord with known zebrafish biology, (2) generate new hypotheses about zebrafish development, and (3) generate new hypotheses on how to improve the experimental procedure.
Again, as our method is agnostic to the tissue types, our trees have no prior assumptions or particular biases regarding the relationships between cell types.

Here we focus on the ordering and relative length of events.
We ignore absolute estimated times since our fitting procedure for a single barcode heavily penalizes large differences between branch lengths.
Though this procedure aids estimation accuracy, it also heavily biases the absolute time estimates.
Thus in the figures, we scale time to be between 0 to $T = 1$ to draw focus away from the absolute times.
(We anticipate the branch length estimates to be more accurate to improve with the GESTALT technology. According to our simulations, the absolute branch length estimates are much more reliable when the several barcodes are inserted into the organism.)

We begin with a coarse summary of the cell lineage tree: We plot the average distance between a leaf node of one tissue type to the most recent ancestor of each different tissue type (Figure~\ref{fig:distance_matrices}).
This matrix recapitulates some well-established facts about zebrafish development.
For example, we estimate that tissue types from the endoderm and mesoderm tended to have shorter shared lineage distances; these tissue types tended to separate from the ectodermal tissues earliest.
This signal potentially captures the migration of the mesoderm and endoderm through the blastopore, isolating them from the ectoderm \citep{SOLNICAKREZEL2005R213}.
In addition, previous studies have found that gills are formed when the anterior part of the intestine grows toward and fuses with the body integument \citep{Shadrin2002-xe}.
The distance matrix shows a large proportion of gill cells dividing late from the other endoderm and mesoderm layers.

The distance matrix also shows that the GFP+ cardiomyocytes tend to be farthest away from other tissue types, which could be either a developmental signal or an artifact of the experimental protocol.
GFP+ cardiomyocytes were sorted using fluorescence-activated cell (FACS) and this purity could drive their separation from the other more heterogeneous organ populations.
An interesting biological speculation would be that the heart is the first organ to form during vertebrate embryo development and, in particular, the myocardial cells are the first to develop, driving this observed signal \citep{Keegan2004-xs}.
These observations show GAPML's improved lineage distance estimation provide a more refined measure of the developmental process, and as our simulations show, will only improve as experimental approaches becomes more sophisticated.

The full cell lineage tree estimated using GAPML for the first adult zebrafish provides significantly more detail than the Camin-Sokal parsimony tree inferred for the original \citet{McKennaaaf7907} publication (Figure~\ref{fig:adr1_tree}).
Our tree has estimated branch lengths whereas the branches were all unit-length in \citet{McKennaaaf7907}.
In addition, the bolded lines in our tree correspond to the caterpillar spines where we have estimated the ordering between children of multifurcating nodes.
Since the original maximum parsimony tree estimate in \citet{McKennaaaf7907} contained many multifurcating nodes and our method converts any multifurcating node to a caterpillar tree, our final tree contains many caterpillar trees.
The longest caterpillar spine in our estimated tree starts from the root node and connects all the major subtrees that share no indel tracts.
As the zebrafish embryo rapidly divides from the single-cell stage, these initial CRISPR editing events establish the founding cell in each subtree.
GAPML estimates the target cut rates to order the events along the caterpillar trees, an impossible task in the original Camin-Sokal multifurcating trees.
Lastly, we observe that the last three subtrees at the end of this spine (farthest away from the root) are primarily composed of alleles only observed in the intestinal bulb and the posterior intestine.
This concords with our understanding of zebrafish development: Of the dissected organs, the digestive tract is the last to fully differentiate at day four \citep{Moody1998-uq}.
In aggregate, these examples again show the power of a refined lineage tree to establish new interesting biological questions and a refined map in which to answer them.

\begin{landscape}
\begin{figure}
\vspace{-2in}
\includegraphics[width=\linewidth]{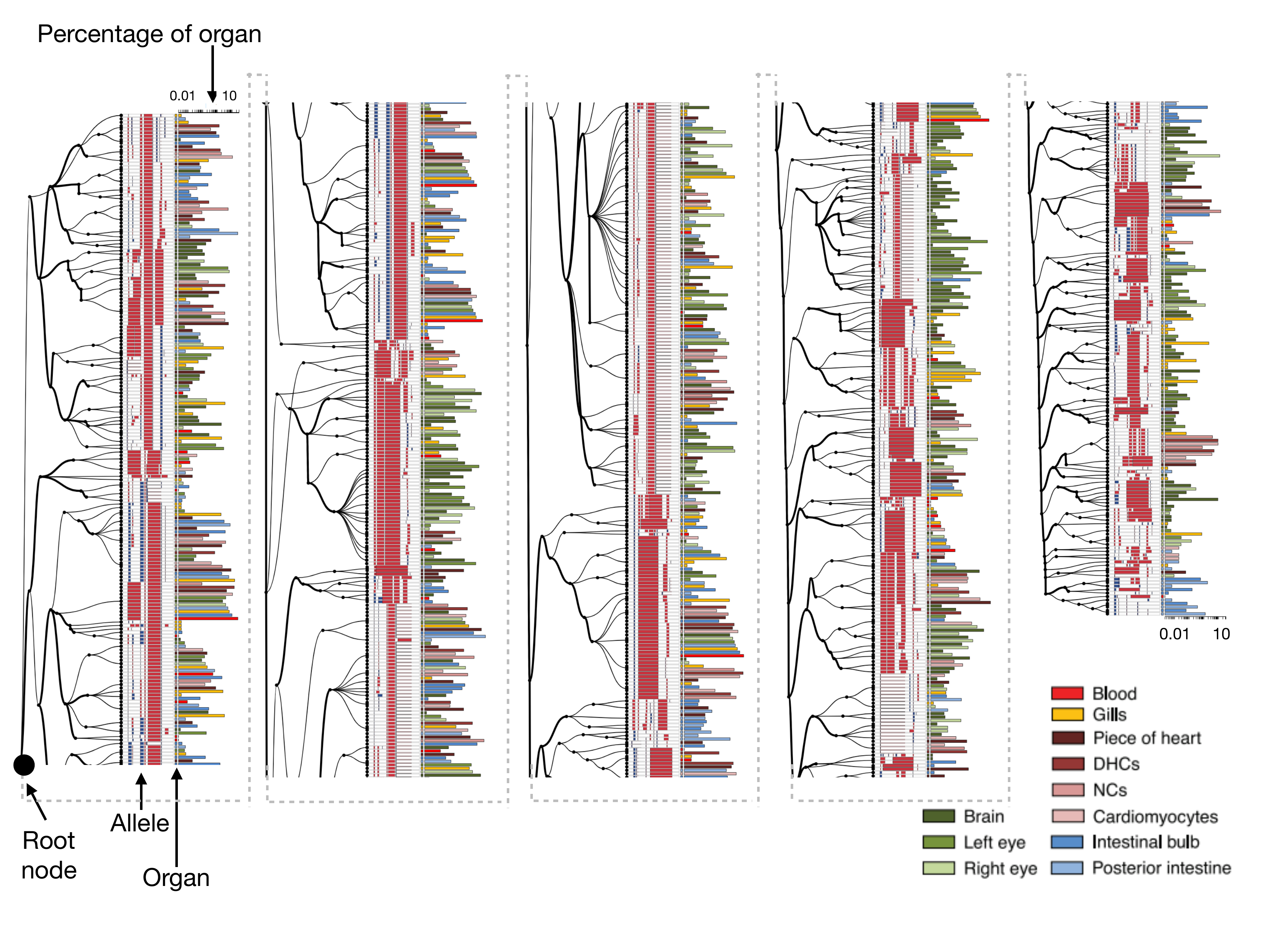}
\caption{
Estimated cell lineage tree for 400 randomly selected alleles from the first adult zebrafish.
Editing patterns in individual alleles are represented as shown previously.
Alleles observed in multiple organs are plotted on separate lines per organ and are connected with stippled branches.
Two sets of bars outside the alleles identify the organ in which the allele was observed and the proportion of cells in that organ represented by that allele (log10 scale).
The bolded lines correspond to the caterpillar spines.
}
\label{fig:adr1_tree}
\end{figure}
\end{landscape}

\subsection{Analysis of GESTALT barcode mutation parameters}

Finally, our model's estimated target cut rate parameters (Table~\ref{table:fitted_params}) provide an interesting resource when considering redesigns of the GESTALT barcode.
Here we focus on the GESTALT barcode in the adult fish.
The estimated target cut rates were very similar across the two fish replicates.

We estimated very different cut rates across the ten targets.
Target 1 and 9 had the highest cut rates; target 3 had the lowest cut rate.
The ratio between the highest and lowest cut rates is at least 10 in both fish, i.e. a deletion at target 1 is at least 10 times more likely to occur than at target 3.
In terms of the tree estimation, the targets with high cut rates mainly help capture early cell divisions whereas targets with low cut rates tend to capture late cell divisions.
Having a broad spectrum of target cut rates is useful for capturing cell divisions throughout the tree, though the specific details depends on the true tree.
Our simulation engine may be useful for understanding how variation in the target rates affects estimation accuracy under various conditions.

The double cut rate is similar across the fish.
The rate is quite high: For the first adult fish, the double cut rate of 0.076 means that the probability of introducing a single-target indel as opposed to an inter-target indel in an unmodified barcode is 59\%.
To counter this, we can decrease the number of long inter-target deletions (and the number of masked events) by placing high cut-rate targets closer together in the barcode.
One potentially helpful barcode design is to place the highest cut-rate targets in the center and the lowest cut-rate targets on the outside.
Alternatively, designers could arrange the targets from highest to lowest cut rate.
Table~\ref{table:fitted_params} shows that the current barcode design is counter to our suggestion, as the two targets with the highest cut rates in the two adult fish are targets 1 and 9.

The characterization of target efficiencies in a compact multi-target barcode is challenging problem.
Our method can help steer the next generation of CRISPR-based lineage recording technologies to have increased recording capacity.

\begin{table}
\begin{center}
\begin{tabular}{lrrrrr}
	\toprule
	{} &   Adult fish \#1 &   Adult fish \#2 &  3 day \#1 &  30 hpf \#5 &  4.3 hpf \#1 \\
	\midrule
	Target 1                      &  3.053 &  1.320 &  0.410 &       1.595 &  1.301 \\
	Target 2                      &  1.232 &  0.317 &  0.155 &       0.697 &  0.474 \\
	Target 3                      &  0.063 &  0.108 &  0.129 &       0.618 &  0.154 \\
	Target 4                      &  1.234 &  0.821 &  0.223 &       0.561 &  0.399 \\
	Target 5                      &  0.619 &  0.542 &  0.179 &       0.510 &  0.276 \\
	Target 6                      &  1.329 &  0.652 &  0.182 &       1.155 &  0.385 \\
	Target 7                      &  0.761 &  0.470 &  0.344 &       0.544 &  1.088 \\
	Target 8                      &  0.090 &  0.136 &  0.141 &       1.171 &  0.151 \\
	Target 9                      &  2.422 &  1.529 &  0.404 &       1.176 &  1.146 \\
	Target 10                     &  0.285 &  0.371 &  0.132 &       1.150 &  0.155 \\
	Double cut rate              &  0.076 &  0.084 &  0.052 &       0.090 &  0.065 \\
	\midrule
	Left trim zero prob   &  0.015 &  0.028 &  0.015 &       0.258 &  0.025 \\
	Left trim length mean  &  6.330 &  6.634 &  5.165 &      12.000 &  6.405 \\
	Left trim length SD    &  4.956 &  5.184 &  4.245 &       6.633 &  4.998 \\
	Right trim zero prob   &  0.906 &  0.834 &  0.890 &       0.238 &  0.818 \\
	Right trim length mean &  4.945 &  3.759 &  3.716 &       4.800 &  3.478 \\
	Right trim length SD   &  6.173 &  5.534 &  5.529 &       4.011 &  5.360 \\
	Insertion zero prob          &  0.400 &  0.411 &  0.401 &       0.520 &  0.419 \\
	Insertion length mean        &  5.085 &  4.540 &  5.786 &       5.446 &  4.589 \\
	Insertion length SD          &  5.798 &  5.533 &  5.219 &       7.291 &  4.708 \\
	\bottomrule
\end{tabular}
\end{center}
\caption{
	Fitted parameters in the adult fish as well as some other fish embryos.
	The parameters above the line are related to target cut rates and the ones below the line are related to the nucleotide deletion and insertion process.
}
\label{table:fitted_params}
\end{table}

\section{Discussion}

In this manuscript, we have proposed a statistical model for the mutation process for GESTALT, a new cell lineage tracing technology that inserts a synthetic barcode composed of CRISPR/Cas9 targets into the embryo.
Our method, GAPML, estimates the cell lineage tree and the mutation parameters from the sequenced modified barcode.
Unlike existing methods, our method estimates branch lengths and the ordering between children nodes that share the same parent.
We demonstrate that our method outperforms existing methods on simulated data, provides more consistent results across biological replicates, and outputs trees that better concord with our understanding of developmental biology.
We have answered the question "Is it possible to reconstruct an accurate cell lineage tree using CRISPR barcodes?'' in \citet{Salvador-Martinez2018-dw} in the affirmative: The cell lineage tree can be estimated to a high degree of accuracy as long as appropriate methods are used.

Our method provides a number of technical contributions to the phylogenetics literature.
The GESTALT mutation process violates many of the classical phylogenetic assumptions that existing methods rely upon for computational tractability.
Thus we determined the most appropriate assumptions that are most suitable in this new setting and developed new methods so that the likelihood is computationally tractable and the estimated trees are accurate.
We believe these techniques could be useful for other phylogenetic problems where the independent-site assumption does not hold.
In addition, our methods may be useful as a jumping off point for analyzing other CRISPR-based cell lineage tracing technologies, such as that using homing CRISPR barcodes \citep{Kalhor2017-ls, Kalhor2018-jb}.
There are still many areas of improvements for the current method, such as quantifying the uncertainty of our estimates, estimating meta-properties about the cell lineage tree for organisms of the same species, and utilizing data sampled at multiple time points.

Finally, the biological results were relatively limited since the goal of this manuscript was primarily on methods development and the data in \citet{McKennaaaf7907} only provide tissue source information.
In future work, we plan to apply our method to analyze data where each allele is paired with much richer information, such as single-cell gene expression data \citep{Raj2018-ux}.

\section{Materials and Methods}
\label{sec:methods}

\subsection{Data}
The data processed in this paper are all from \citet{McKennaaaf7907} and are available at the Gene Expression Omnibus under GSE81713.
We use the aligned data where each allele was described with the observed insertion/deletions (indels) at each target.
Each CRISPR target can only be modified once and indels can only be introduced via a double-stranded break at a target cut site, so we further processed the aligned data accordingly: merging indels if there were more than one associated with a given target, and extending the deletion lengths and insertion sequences so that a target cut site was nested within each indel.
For this paper, we assume that the processed data is correct and do not attempt to model the effects of alignment error.

\subsection{Methodological contributions}
Here we highlight the key methodological contributions that we needed to develop in order to analyze GESTALT.
We needed to develop new methodology because the GESTALT mutation process violates many classical assumptions in phylogenetics.

Since GESTALT is a new technology, we introduce new mathematical abstractions for the biological process.
We then consider all statistical models that satisfy our proposed set of assumptions.
We carefully designed these assumptions to balance biologically realism and computationally feasibility.
We next show that such models are ``lumpable'' and use this to efficiently calculate the likelihood.
Even though lumpability has been used to reduce computation time for general Markov chains, this idea is rarely used in phylogenetics.
We show how lumpability can be combined with Felsenstein's algorithm if the mutation process is irreversible.

In addition, our method estimates trees at a finer resolution compared to other methods, which leads to better tree estimates.
In particular, we resolve the multifurcations as a caterpillar tree and show how to efficiently search over caterpillar tree topologies.
As far as we know, this is one of the few methods that tunes the tree topology, which is typically treated as a combinatorial optimization problem, by solving a continuous optimization problem.
The closest equivalent appears in the Bayesian phylogenetics literature where local changes to the topology may be introduced via a Subtree-Slide operator \citep{Hohna2008-jm}.

To handle the small number of barcodes, we improve the stability of the method by maximizing the \textit{penalized} log likelihood.
Previous phylogenetic methods that penalize branch lengths assume that the tree topology is known \citep{Kim2008-rk, Zhang2018-ky}.
However, the tree topology is unknown in GESTALT.
We show how penalization can be combined with tree topology search methods.
Combining the two is not trivial since a na\"ive approach will bias the search towards incorrect topologies.

Finally, we use an automatic differentiation framework to optimize the phylogenetic likelihood.
Automatic differentiation software has accelerated deep learning research since they allow researchers to quickly iterate on their models \citep{Baydin2018-py}.
Likewise, we found that this tool greatly accelerated our progress and, through this experiment, we believe that this tool may greatly accelerate the development of maximum likelihood estimation methods in phylogenetics.

\subsection{GESTALT framework and definitions}
In this section, we present mathematical definitions for the many components in GESTALT, though begin in this paragraph by giving an overview in words.
We begin with defining the \textit{barcode} and the individual \textit{targets} within it.
A barcode is mutated when nucleotides are inserted and/or deleted, which is referred to as an \textit{indel tract}.
We then define a possibly-mutated barcode or \textit{allele} as a collection of the observed indel tracts.
Finally, we use all these abstractions to define the barcode mutation process, which is framed as stochastic process where the state space corresponds to all possible alleles and the transitions correspond to indel tracts.
To aid the reader, Table~\ref{table:notation_reference} briefly summarizes the definitions used in the paper.

The unmodified barcode is a nucleotide sequence where $M$ disjoint subsequences are designated as \underline{targets} (Figure~\ref{fig:barcode}).
The targets are numbered from 1 to $M$ from left to right and the positions spanned by target $j$ are specified by the set $\pos(j)$.
Each target $j$ is associated with a single cut site $c(j) \in \pos(j)$.

\begin{figure}
\begin{subfigure}{0.45\textwidth}
\includegraphics[width=\linewidth]{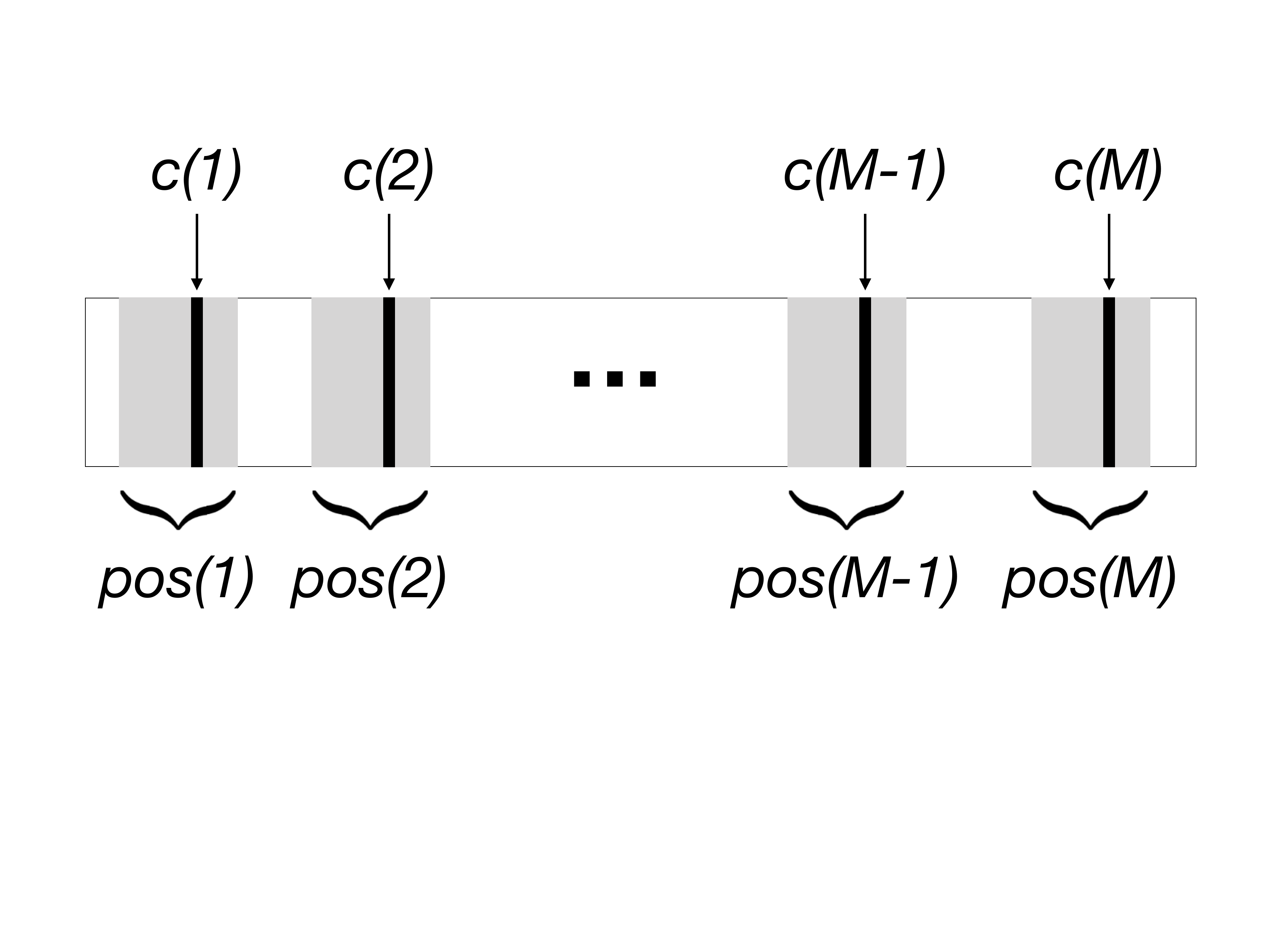}
\caption{
A barcode with $M$ targets.
The cut site of the targets, $c(j)$ for $j = 1,...,M$, are shown by the bolded black lines.
The positions associated with each target are highlighted in the gray boxes.
}
\label{fig:barcode}
\end{subfigure}
\hspace{0.01in}
\begin{subfigure}{0.45\textwidth}
\includegraphics[width=\linewidth]{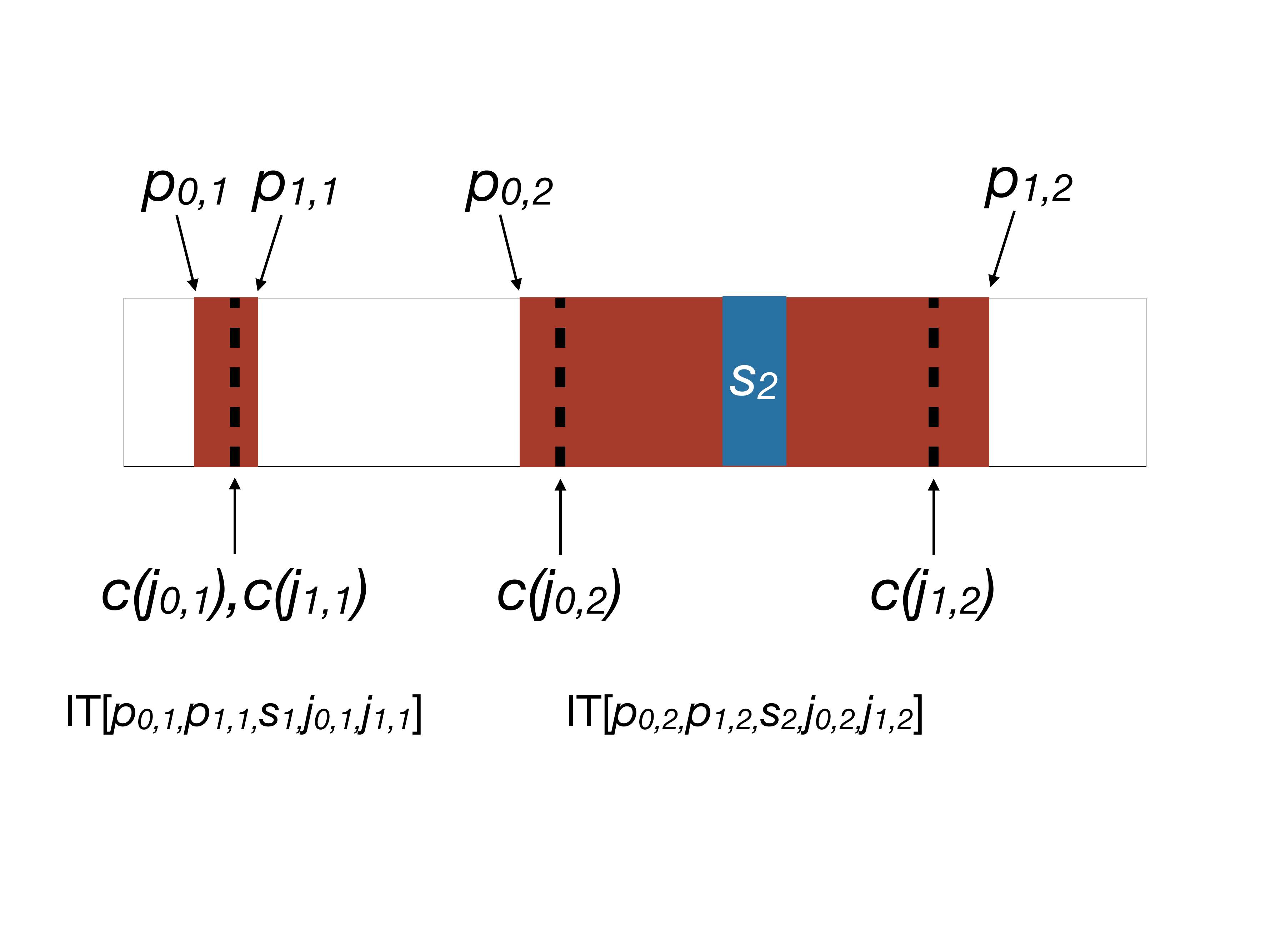}
\caption{
An example allele with two indel tracts.
The left indel tract was introduced by a cut at a single target and does not have an insertion, i.e. $s_1 = \emptyset$.
The right indel tract was introduced by cuts at two targets.
Red indicates deletion and blue indicates insertion.
}
\label{fig:allele}
\end{subfigure}
\caption{Illustration of GESTALT definitions}
\label{fig:barcode_math}
\end{figure}

A barcode can be modified by the introduction of an \underline{indel tract}.
An indel tract, denoted by $\IT[p_0, p_1, s, j_0, j_1]$, is a mutation event in which targets $j_0$ and $j_1$ are cut ($j_0 \le j_1$), positions $p_0, p_0 + 1,..., p_1 - 1$ in the unmodified barcode are deleted, and a nucleotide sequence $s$ is inserted.
If $j_0 = j_1$, only a single target is cut.
When $p_0 = p_1$, then no positions are deleted.
If $s$ is of length zero, then no nucleotides are inserted.
We only consider indel tracts that modify the sequence, i.e. either $p_0 < p_1$ or $s$ has positive length, and nest at least one target cut site between positions $p_0$ and $p_1$.

A possibly-modified barcode, also referred to as an \underline{allele}, is a sequence of disjoint indel tracts associated with a single barcode (Figure~\ref{fig:allele})
\begin{align}
a \equiv \left \{\IT[p_{0,k}, p_{1,k}, s_k, j_{0,k}, j_{1,k}]: k \in \{1,...,m\} \right \}
\label{eq:allele_def}
\end{align}
where $m \ge 0$ and $p_{1,k} < p_{0,k+1}$ for $k = 1,...,m - 1$.
The positions $p_{0,k}, p_{1,k}$ in the indel tracts always refer to the positions in the original unmodified barcode: deletions and insertions do not change the indexing.
Let $\Omega$ be the set of all possible alleles.

A target $j$ is active in allele $a$ if none of the nucleotides in $\pos(j)$ are modified.
That is, the \underline{status} of target $j$ in allele $a$ is
$$
\TargStat(j;a) = \mathbbm{1}\{\exists \IT[p_0, p_1 ,s, j^{'}_0, j^{'}_1] \in a \text{ and } \exists k \in \pos(j) \text{ s.t. } p_0 \le k \le p_1 \}.
$$
So $\TargStat(j;a)$ is 0 if target $j$ is active and 1 if it is inactive.
For convenience, denote the target status of allele $a$ as
\begin{align}
\TargStat(a) = (\TargStat(1;a), ..., \TargStat(M;a)).
\label{eq:targ_stat}
\end{align}

We now introduce the rules governing how alleles change through the introduction of indel tracts.
First, transitions between possible allele states $\Omega$ are constrained by the status of the targets.
For a given allele $a$, we can introduce the indel tract $d = \IT[p_0, p_1, s, j_0, j_1]$ if and only if targets $j_0$ and $j_1$ are active.
Note that the set of transitions allowed under this rule is a superset of biologically-plausible transitions.
For example, even if position $p_0$ is deleted, introducing indel tract $d$ is formally allowed in our scheme.
However in order to exclude biologically-implausible transitions, our models assign near-zero probability to such transitions.

Let $\Apply(a, d)$ be the resulting allele from introducing indel tract $d$ into allele $a$.
If indel tract $d$ does not overlap any other indel tract in $a$, then $\Apply(a, d)$ is simply the union $a \cup \{d\}$.
If $d$ completely masks other indel tracts, i.e. all indel tracts in $a$ are either completely within the range of $p_0$ to $p_1 - 1$ or are completely outside of the range, then $\Apply(a, d)$ is the resulting allele after removing the masked indel tracts and adding $d$.
The last possibility is that $d$ is adjacent or overlaps, but does not fully mask, other indel tract(s) in $a$; then $\Apply(a, d)$ is the resulting allele after properly merging overlapping indel tracts (Figure~\ref{fig:overlap_indels}).
From a biological perspective, it is impossible to introduce overlapping but non-masking indel tracts.
However the likelihood model is much simpler if we allow such events to happen.
Since deletion lengths tend to be short in the empirical data (75\% quantile = 10, \citet{McKennaaaf7907}), we will estimate that long deletion lengths occur with small probability, which means that these overlapping but non-masking indel tracts have very small probabilities.
Thus we believe our model closely approximates the GESTALT mutation process.
We will also discuss an assumption that removes these overlapping indel tracts from the likelihood calculation entirely in more detail in the following section.
\begin{figure}
\begin{center}
\includegraphics[width=0.4\textwidth]{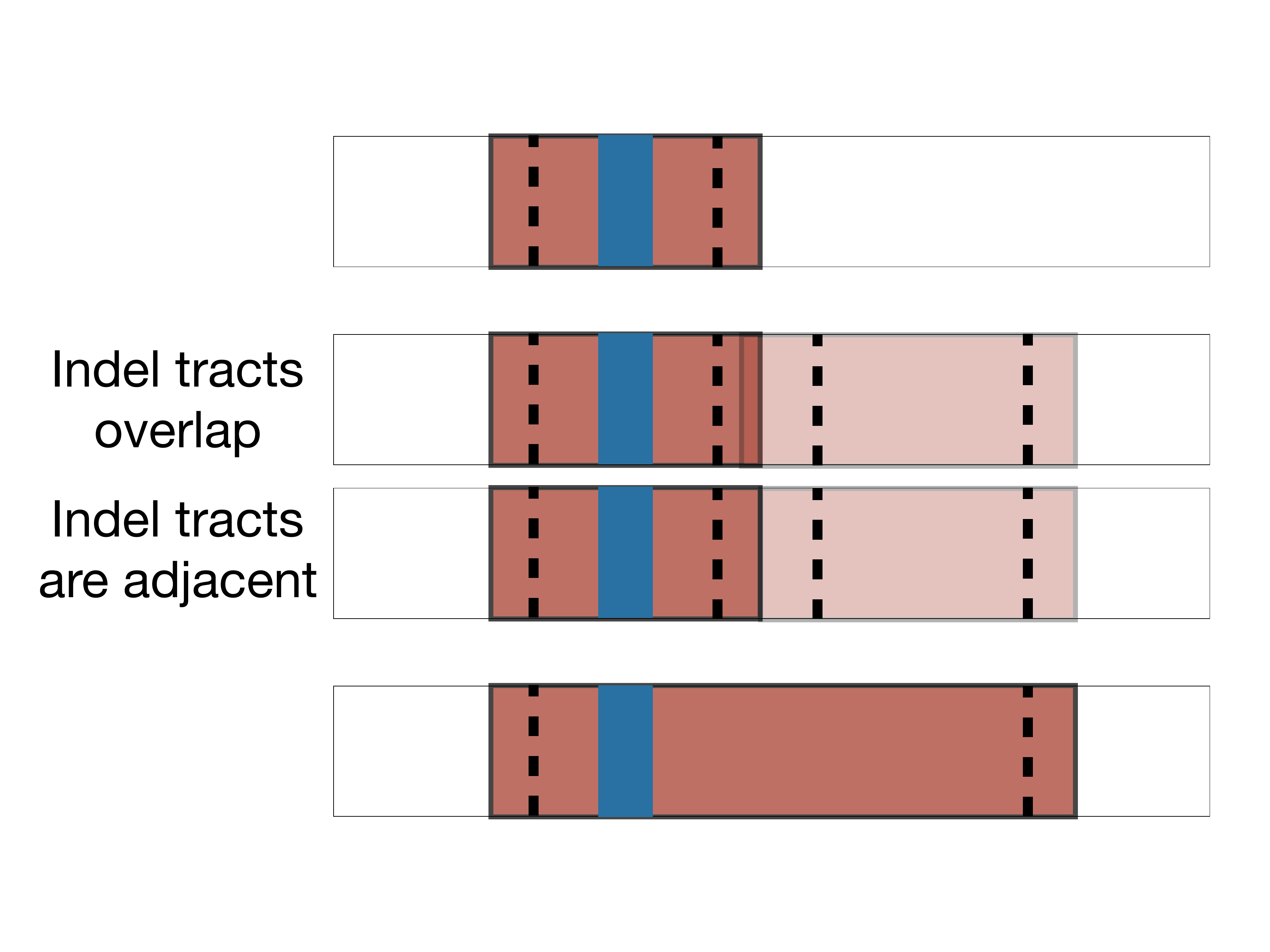}
\end{center}
\caption{
	Possible transitions from the top allele and the bottom allele:
	either an indel tract is introduced that overlaps with an existing indel tract (middle top) or an adjacent indel tract is introduced (middle bottom).
}
\label{fig:overlap_indels}
\end{figure}

\subsection{GESTALT Model and Assumptions}

\begin{table}
\begin{center}
\begin{tabular}{c|p{8cm}|c}
Symbol & Description & Eq. \\
\midrule
$X_{\N}(t)$ & Markov process along branch ending with node $\N$&\\
$X_{S}(t)$ & Markov process along branches ending with nodes in set $S$&\\
$a_{\N}$ & Observed allele at leaf $\N$&\\
$ a_{S}$ & Observed alleles at leaves in set $S$&\\
$\pos(j)$ & Positions of target $j$ in the unmodified barcode&\\
$c(j)$ & Cut site of target $j$&\\
$\Leaves(\N)$ & Leaves of node $\N$&\\
$\Desc(\N)$ & Descendants of node $\N$ &\\
$\TargStat(a)$ & Status of targets in allele $a$. 1 in position $j$ indicates target $j$ is inactive & \eqref{eq:targ_stat}\\
$\IT[p_0, p_1, s, j_0, j_1]$ & Indel tract that cuts targets $j_0$ and $j_1$, deletes positions $p_0$ to $p_1 - 1$, inserts $s$&\\
$\TT[j_0', j_0, j_1, j_1']$ & Target tract: all indel tracts that cut targets $j_0$ and $j_1$ and deactivate targets $j_0'$ to $j_1'$& \eqref{eq:tt} \\
$\TT(d)$ & The target tract that indel tract $d$ belongs to & \\
$\WC[j_0, j_1]$ & Wildcard: any indel tract that only deactivates targets with indices $j_0$ to $j_1$ & \eqref{eq:wc}\\
$\SGWC[p_0, p_1, s, j_0, j_1]$ & Singleton-wildcard: union of an indel tract and its inner wildcard & \eqref{eq:sgwc}\\
$\AncState(\N)$ & The set of likely ancestral states for node $\N$ & \eqref{eq:anc_state}
\end{tabular}
\end{center}
\caption{Notation used in this paper}
\label{table:notation_reference}
\end{table}

\begin{figure}
\begin{center}
\includegraphics[width=0.8\textwidth]{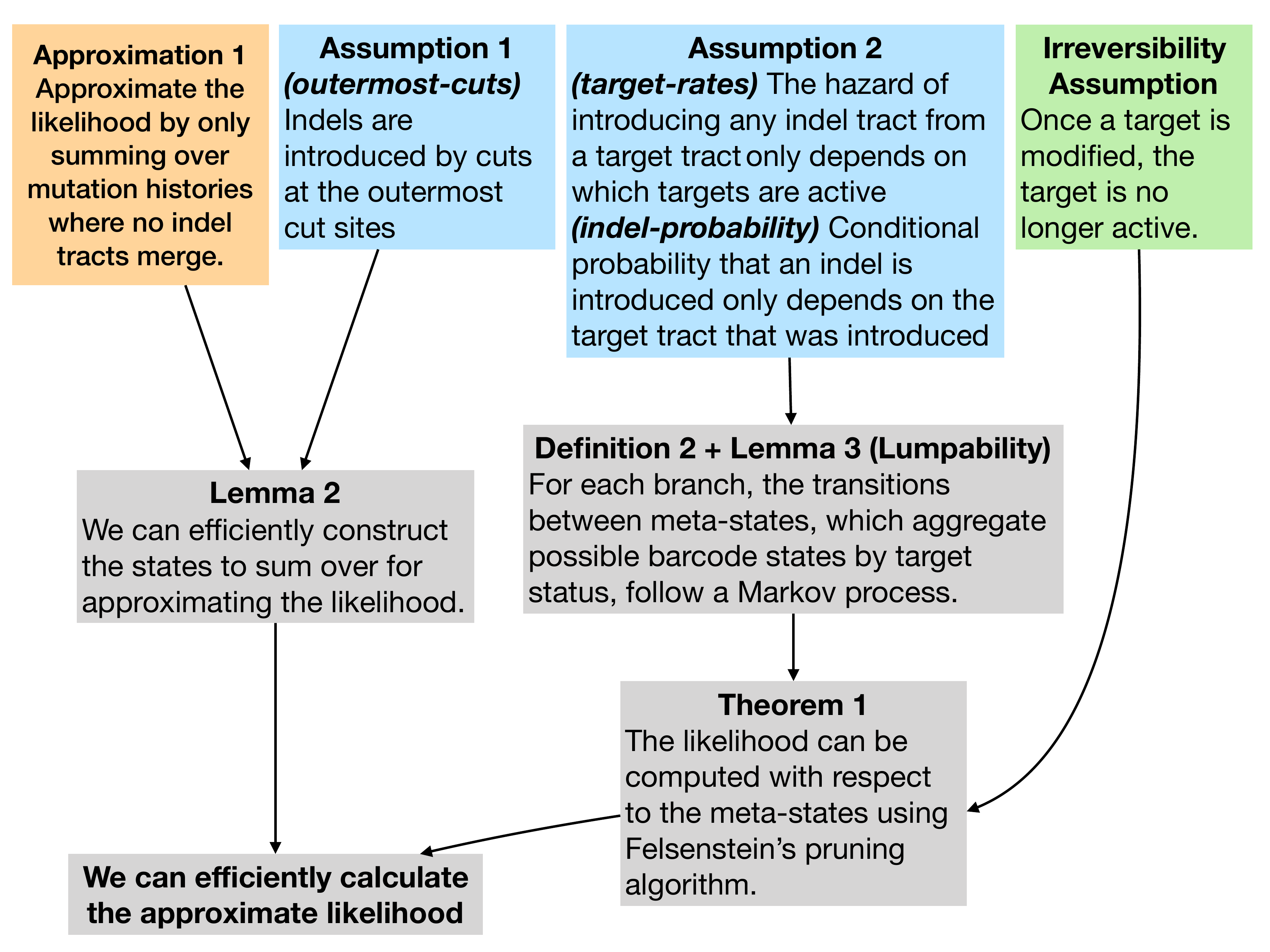}
\end{center}
\caption{
A guide for how assumptions, approximations, and derived results connect and lead to our final algorithm for approximating the likelihood.
The flowchart also maps the formal assumptions here to Assumptions A, B, and C in the introduction.}
\label{fig:math_flow}
\end{figure}

Here we define the GESTALT model more concretely and formalize the model assumptions presented in the Introduction.
Recall the three model assumptions.
First, as we rarely observe very long deletion lengths in the GESTALT data, we assume that indel tracts are introduced by cuts at the outermost cut sites (assumption A).
In addition, based on our biological understanding of the CRISPR/Cas9 mutation process, we assume that mutations occur in two stages: first the targets are cut according to the rates of the active targets (assumption B) and then nucleotides are deleted/inserted according to a process that only depends on which targets were cut (assumption C).
Figure~\ref{fig:math_flow} presents a flowchart of how the assumptions are used to derive later results.

The barcode mutation process up to time $T$ is formulated as a continuous time Markov chain $\{X(t): 0 \le t \le T\}$ with state space $\Omega$.
As mentioned before, the choice of using state space $\Omega$ implicitly assumes that indel tracts are introduced instantaneously, i.e. nucleotides are inserted and/or deleted immediately after target(s) are cut.

For a given tree $\mathbb{T}$, let $t_{\N}$ be the length of the branch ending with node $\N$ and let $\{X_{\N}(t): 0 \le t_{\N}\}$ be the Markov process along the branch.
In addition, let $a_{\N}$ be the allele observed at that leaf, and let $\Leaves(\N)$ denote all leaves with ancestral node $\N$.
The set of leaves in the entire tree $\mathbb{T}$ is denoted $\Leaves(\mathbb{T})$.
As notational shorthand, we denote the Markov process for the branches with end nodes  in the set $S$ as $X_{S}$.
In addition, the observed alleles in the leaf set $S$ are denoted $a_{S}$.

If there are multiple barcodes, we use the notation $X_{\N}^{(i)}(\cdot)$ to represent the process for the $i$th barcode and $a_{\N}^{(i)}$ to represent the allele observed on the $i$th barcode.
In the manuscript, we assume the barcodes are independently and identically distributed.
Therefore we will typically discuss the model and assumptions in the context of a single barcode and omit the index of the barcode.

Unfortunately, calculating the likelihood of the tree for a general model where the mutation rate depends on the entire sequence is computationally intractable: We would need to compute the transition rates between an infinite number of barcode states.
Instead we introduce a model assumption so that we can aggregate the possible barcode states into lumped states indexed by target statuses.
Then we can compute the likelihood using transition matrices of dimension at most $2^M \times 2^M$.
As $M$ is typically small ($M = 10$ in \citet{McKennaaaf7907}), the assumption makes the likelihood computationally feasible.

First we formalize the \emph{outermost-cuts} assumption, which states that the cuts for an indel tract occur at the outermost cut sites.
We define this mathematically by requiring that for any indel tract where targets $j_0$ and $j_1$ are cut, the deletions to the left and right cannot extend past the cut site of the neighboring targets $j_0 - 1$ and $j_1 + 1$.
\begin{assumption}[\emph{outermost-cuts}]
	All indel tracts are of the form $\IT[p_0, p_1, s, j_0, j_1]$ where
	\begin{align*}
	c(j_0 - 1) < p_0 \le c(j_0) & \quad \text{ if } j_0 \ge 1\\
	c(j_1) \le p_1 < c(j_1 + 1) & \quad  \text{ if } j_1 \le M.
	\end{align*}
	\label{assump:deletion_not_too_long}
\end{assumption}
\noindent This assumption limits the possible mutation histories of the alleles.
Note that Assumption~\ref{assump:deletion_not_too_long} still allows indel tracts to deactivate targets immediately neighboring the cut site.
That is, an indel tract $d = \IT[p_0, p_1, s, j_0, j_1]$ can either have a \textit{short} deletion to the left so that target $j_0 - 1$ is unaffected or a \textit{long} deletion to the left such that target $j_0 - 1$ is deactivated, i.e.
\begin{align}
d \text{ has a short left deletion if } & p_0 \in \pos(j_0)
\label{eq:short_long1}
\\
d \text{ has a long left deletion if } & p_0 \in \pos(j_0 - 1)
\label{eq:short_long2}
\end{align}
We can have similar short and long deletions to the right.

For the second assumption, let us introduce the concept of a \underline{target tract}, which is a set of indel tracts that cut the same target(s) and deactivate the same target(s).
A target tract, denoted $\TT[j_0', j_0, j_1,  j_1']$, is the set of all indel tracts that cut targets $j_0$ and $j_1$ and delete nucleotides such that targets $j_0'$ through $j_1'$ are inactive, i.e.
\begin{align}
\TT[ j_0', j_0, j_1, j_1'] = \left\{
\IT[p_0, p_1, s, j_0, j_1]: p_0 \in \pos(j_0'), p_1 \in \pos(j_1')
\right\}.
\label{eq:tt}
\end{align}
Note that we always have that $j_0' \le j_0 \le j_1 \le j_1'$.
Each indel tract $d$ belongs to a single target tract; we denote its associated target tract by $\TT(d)$.

This second assumption decomposes the mutation process into a two-step process where targets are cut and then indels are introduced; and combines the \emph{target-rate} and \emph{indel-probability} assumptions.
In particular, the assumption states that the instantaneous rate of introducing an indel tract can be factorized into the rate of introducing any indel tract from a target tract, which depends on the target status of the current allele, and the conditional probability of introducing an indel tract which only depends on the target tract.
\begin{assumption}[\emph{rate-rate, indel-probability}]
	Let $a$ be any allele, $d$ be any indel tract that can be introduced into $a$, and $\tau$ be the target tract $\TT(d)$.
	The instantaneous rate of introducing indel tract $d$ in allele $X(t) = a$ can be factorized into two terms: first, a function $h$ that only depends on $\tau$, $\TargStat(a)$, and time $t$, then second, the conditional probability of introducing $d$ given $\tau$:
	\begin{align*}
    \begin{split}
	q(a, \Apply(a, d), t)
    & := \lim_{\Delta \rightarrow 0}
	\frac{\Pr \left (X(t + \Delta) = \Apply(a, d) | X(t) = a \right )}{\Delta} \\
	& = h \left (\tau, \TargStat(a), t \right )
	\Pr \left (d | \tau \right).
    \end{split}
	\end{align*}
	\label{assump:factor_rate}
\end{assumption}
\noindent Thus the instantaneous rate of introducing $d$ only depends the allele $a$ through its target status.
Using this assumption, we will show that the mutation process is equivalent to a continuous time Markov chain where we lump together possible allele states that share the same target status.

\subsection{Likelihood approximation: summing over likely ancestral states}

The likelihood of a given tree and mutation parameters is the sum of the probability of the data over all possible mutation histories.
There are many possible ancestral states since inter-target deletions can mask previously-introduced indel tracts.
In this section, we present an approximation of the likelihood that only sums over the likely ancestral states and ignores those with very small probabilities.
We also provide a simple algorithm that efficiently specifies the set of these likely states, which is useful when we actually implement the (approximate) likelihood calculation.

We first address the problem that inter-target indel tracts have too many possible histories for brute-force computation.
Not only can inter-target deletions mask previously-introduced indel tracts, but they can also result from merging overlapping multiple indel tracts.
To simplify the likelihood calculation, we ignore any mutation history where indel tracts merge.
We are motivated to do so by observing that indel tracts rarely merge -- merging only occurs when many nucleotides are deleted whereas we mostly observe short deletions in the experimental data.
Thus, mutation histories involving merging indel tracts probably contribute a negligible fraction to the total likelihood.
The approximation is therefore as follows:
\begin{approximation}
We approximate the likelihood with the probability of the data over all possible mutation histories where no two indel tracts are ever merged:
\begin{align}
\Pr\left( X_{\Leaves(\N)}(T) = a_{\Leaves(\N)} \right)
\approx \Pr \left ( X_{\Leaves(\N)}(T) = a_{\Leaves(\N)}, \text{no indel tracts merged} \right ).
\label{eq:approx}
\end{align}
\label{approx:no_merging}
\end{approximation}
\noindent Note that this approximation strictly lower bounds the full likelihood since we sum over a subset of the possible histories.

Based on Approximation~\ref{approx:no_merging}, we define a \underline{partial ordering between alleles}.
Given two alleles $a, a' \in \Omega$, we use $a' \succeq a$ to indicate that $a$ can transition to the allele $a'$ without merging indel tracts, i.e. there is a sequence of indel tracts $\{d_i\}_{i=1}^m$ such that
$$
a' = \Apply \left (d_m, \Apply \left (d_{m-1}, ....\Apply (d_1, a) \right ) \right )
$$
where no indel tracts merge.

Now we present an algorithm to concisely express the set of alleles that are summed over at each internal node under Approximation~\ref{approx:no_merging}.
For the tree $\mathbb{T}$ and observed alleles $a_{\Leaves(\mathbb{T})}$, define the set of likely ancestral states at each internal node $\N$ as
\begin{align}
\AncState\left (\N; \mathbb{T}, a_{\Leaves(\mathbb{T})} \right )
= \left \{
a \in \Omega:
a \preceq a_{\LL} \forall \LL \in \Leaves(\N)
\right \}.
\label{eq:anc_state}
\end{align}
We define $\AncState(\cdot)$ over leaf nodes in the same way as \eqref{eq:anc_state}, but we interpret this set as the alleles that likely preceded the observed allele.
Henceforth, we use the shorthand notation $\AncState(\N) \equiv \AncState\left (\N; \mathbb{T}, a_{\Leaves(\mathbb{T})} \right )$ whenever the context is clear.
The approximate likelihood from Approximation~\ref{approx:no_merging} is equal to summing over $\AncState(\N)$ at each internal node $\N$.

To construct the sets $\AncState(\N)$, we also need sets of indel tracts of the following forms (Figure~\ref{fig:sgwc}):
\begin{itemize}
	\item \underline{wildcard}\footnote{In software systems, a wildcard is a symbol used to represent one or more characters (e.g. ``*''). Similarly, we define wildcard here as all indel tracts that only deactivate targets within a specified range.} $\WC[j_0, j_1]$: the set of all indel tracts that only deactivate targets within the range $j_0$ to $j_1$
	\begin{align}
	\WC[j_0, j_1] = \left\{
	\IT[p_0', p_1', s', j_0', j_1']: \pos(j_0 - 1) < p_0',  p_1' < \pos(j_1 + 1)
	\right\}
	\label{eq:wc}
	\end{align}
	\item \underline{singleton-wildcard} $\SGWC[p_0, p_1, s, j_0, j_1]$: the union of the singleton set $\{\IT[p_0, p_1, s, j_0, j_1]\}$ and its inner wildcard, if it exists:
	\begin{align}
    \begin{split}
	& \SGWC[p_0, p_1, s, j_0, j_1] = \\
    & \qquad
	\begin{cases}
	\{\IT[p_0, p_1, s, j_0, j_1]\} \cup \WC[j_0 + 1, j_1 - 1] & \text{if } j_0 + 1 \le j_1 - 1\\
	\{\IT[p_0, p_1, s, j_0, j_1]\} & \ow
	\end{cases}
    \end{split}
	\label{eq:sgwc}
	\end{align}
	The singleton of singleton-wildcard \eqref{eq:sgwc} refers to $\{\IT[p_0, p_1, s, j_0, j_1]\}$ and the inner wildcard of a singleton-wildcard refers to $\WC[j_0 + 1, j_1 - 1]$ if it exists and $\emptyset$ otherwise.
\end{itemize}
Two singleton-wildcards are disjoint if the range of their positions don't overlap.
A wildcard is disjoint from a wildcard/singleton-wildcard if the range of their targets don't overlap.
\begin{figure}
	\begin{subfigure}{0.6\linewidth}
		\includegraphics[width=\linewidth]{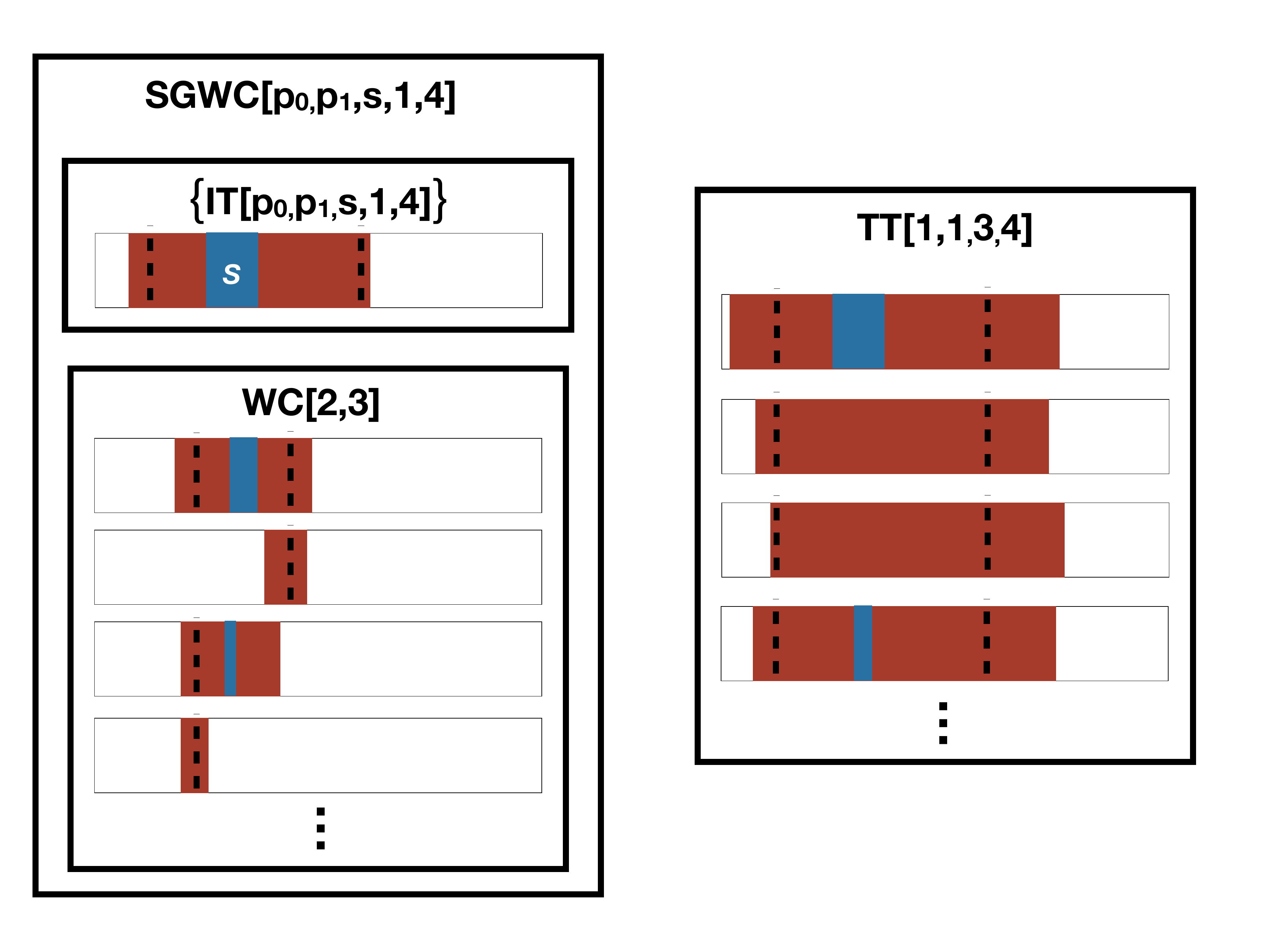}
		\caption{
		Relationship between indel tracts ($\IT$), target tracts ($\TT$), wildcards ($\WC$), and singleton-wildcards ($\SGWC$).
		Each indel tract is shown in the context of a barcode.
		Each box represents a set of indel tracts; we show the notation at the top of the box for describing that set of indel tracts.
		A singleton-wildcard is the union of a singleton set composed of an indel tract and an inner wildcard.
		}
		\label{fig:sgwc}
	\end{subfigure}
	\hspace{0.02in}
	\begin{subfigure}{0.38\linewidth}
		\includegraphics[width=\linewidth]{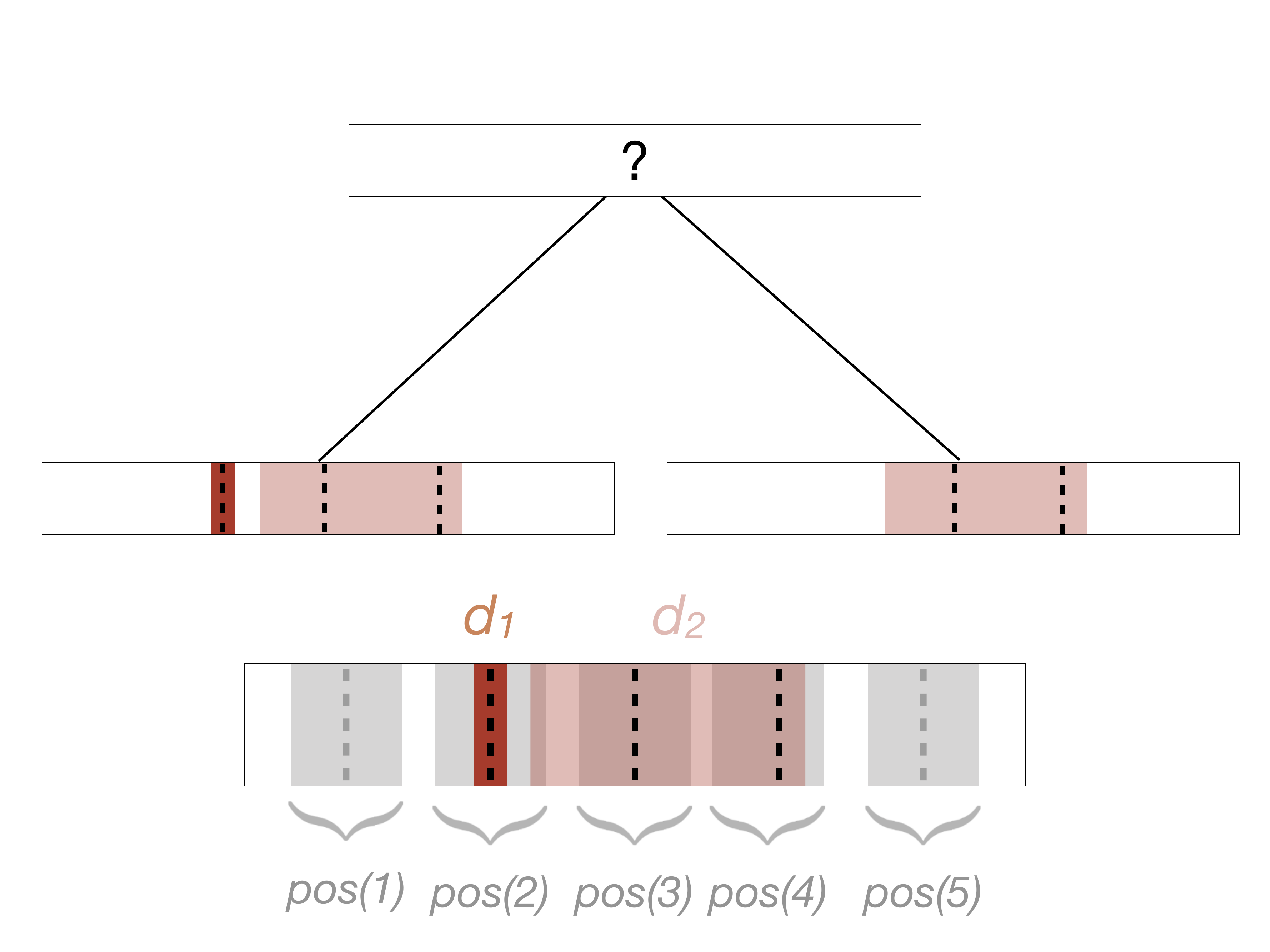}
		\caption{
			\textbf{Bottom}: An example of an allele with two indel tracts $d_1$ and $d_2$ where $d_1$ must have been introduced before $d_2$, because $d_1$ cuts target 2 while $d_2$ deactivates target 2 through 4.
			\textbf{Top}: A two-leaf tree where one leaf is the example allele and the other leaf is an allele with only $d_2$.
			Since $d_1$ must be introduced before $d_2$, the only possible ancestral state of the parent is an unmodified allele.
			On the other hand, if $d_2$ did not overlap with $\pos(2)$, we can simply take the intersection of the two alleles to get a possible ancestral state.
		}
		\label{fig:indel_tract_order}
	\end{subfigure}
\caption{Visual dictionary of indel tract sets (a) and example of ordering between indel tracts in an allele (b).}
\label{fig:indel_tract_relations}
\end{figure}

Given a set of indel tracts $D$, let the alleles generated by $D$, denoted $\Alleles(D)$, be the set of alleles that can be created using subsets of $D$, i.e. $\Alleles(D)$ is
\begin{align*}
\left\{
\{\IT[p_{0,k}, p_{1,k}, s_j, j_{0,k}, j_{1,k}]\}_{k=1}^m \subseteq D: m \in \mathbb{N}, p_{1,k} < p_{0,k+1} \forall k = 1,...,m - 1
\right\}.
\end{align*}
We are interested in wildcards and singleton-wildcards because for any leaf node $\LL$ with indel tracts $\IT[p_{0,m}, p_{1,m}, s_m, j_{0,m}, j_{1,m}]$ for $ m = 1,...,M_{\LL}$, a superset of $\AncState(\LL)$ is the alleles generated by the union of their corresponding singleton-wildcards, i.e.
\begin{align}
\AncState(\LL)
\subseteq
\Alleles \left(
\bigcup_{m = 1,...,M_{\LL}} \SGWC[p_{0,m}, p_{1,m}, s_m, j_{0,m}, j_{1,m}]
\right).
\label{eq:leaf_superset}
\end{align}

In fact, the following lemma states that we can use a recursive algorithm to compute supersets of $\AncState(\N)$.
The algorithm starts at the leaves and proceeds up towards the root.
We first use \eqref{eq:leaf_superset} to compute supersets of $\AncState(\LL)$ for all leaf nodes $\LL$.
Now consider any internal node $\N$ with two (direct) children nodes $\C_1$ and $\C_2$.
For mathematical induction, suppose that we have already computed supersets of $\AncState(\C_i)$ for $i = 1$ and $2$ that are the alleles generated by unions of wildcards/singleton-wildcards.
We compute the superset of $\AncState(\N)$ as the alleles generated by the nonempty pairwise intersections of wildcards and/or singleton-wildcards corresponding to $\AncState(\C_1)$ and $\AncState(\C_2)$.
Since the intersection of wildcards and/or singleton-wildcards is always a wildcard or singleton-wildcard, we can also write this superset as the alleles generated by the union of disjoint wildcards and/or singleton-wildcards.
We can repeatedly intersect wildcards/singleton-wildcards for internal nodes with multiple children nodes.
The proof for the lemma is straightforward so we omit it here.
\begin{lemma}
	\label{lemma:anc_state_superset}
	Consider any internal node $\N$ with children nodes $\C_1$, ..., $\C_K$.
	Suppose for each child $\C_k$, we have that
	\begin{align}
	\AncState(\C_k)
	&\subseteq \Alleles\left(
	\bigcup_{m=1}^{M_{\C_k}} D_{\C_k, m}
	\right)
	\label{eq:anc_state_superset}
	\end{align}
	where $\left\{
	D_{\C_k, m}
	\right\}_{m = 1}^{M_{\C_k}}$ are pairwise disjoint wildcards and/or singleton-wildcards.
	Then, we can also write $\AncState(\N)$ in the form of \eqref{eq:anc_state_superset} where $\left\{
	D_{\N, m}
	\right\}_{m = 1}^{M_{\N}}$ are the non-empty intersections of $D_{\C_1, m_1} \cap ... \cap D_{\C_K, m_K}$, i.e.
	\begin{align}
	\left\{
	D_{\N, m}
	\right\}_{m = 1}^{M_{\N}}
	& = \left\{
	D' = D_{\C_1, i_1} \cap ... \cap D_{\C_K, m_K}:
	D' \ne \emptyset,
	m_1 \in \{ 1,\cdots,M_{\C_1}\},
	\cdots,
	m_K \in \{ 1,\cdots,M_{\C_K}\}
	\right \}.
	\label{eq:intersect_sgwc_nonempty}
	\end{align}
\end{lemma}

For efficiency reasons, we are not satisfied with computing supersets of $\AncState(\cdot)$; rather, we would like to concisely express the set of alleles that is exactly equal to $\AncState(\cdot)$.
The only case in which the algorithm computes a strict superset of $\AncState(\N)$ is when the alelles observed at the leaves of $\N$ imply that the observed indel tracts must be introduced in a particular order.
For example, if an allele has indel tracts $d_1$ and $d_2$, we know that $d_1$ must be introduced before $d_2$ if $d_1$ cuts target $j$ and $d_2$ deactivates target $j$ (Figure~\ref{fig:indel_tract_order}).
Due to this ordering, we may find that the same indel tract observed in two alleles must have been introduced independently (also known as homoplasy in the phylogenetics literature).
To indicate such orderings, we use the notation $d \in a \Rightarrow d' \in a$ to denote that ``if indel tract $d$ is in allele $a$, then indel tract $d'$ must also be in $a$.``
The set of alleles respecting this ordering constraint is denoted
\begin{align*}
\Order\left(
d \Rightarrow d'
\right )
=
\left\{
a \in \Omega:
d \in a \implies d' \in a
\right\}.
\end{align*}
Per this definition, $\Order(d \Rightarrow d')$ contains all alleles that do not include $d$.

The following lemma builds on Lemma~\ref{lemma:anc_state_superset} and computes the sets exactly equal to $\AncState(\cdot)$.
The algorithm is similar as before, but the parent nodes also adopt ordering requirements from their children nodes.
Note that ordering requirements only ever involve observed indel tracts.
Again, the proof for the lemma is straightforward so we omit it here.
\begin{lemma}
	\label{lemma:anc_state_summary}
	For any leaf node $\LL$, suppose its observed allele is $\{d_m: m = 1,...,M_{\LL}\}$ for some $M_{\LL} \ge 0$, where $d_m = \IT[p_{0,m}, p_{1,m}, s, j_{0,m}, j_{1,m}]$.
	Denote its list of ordering requirements as
	\begin{align*}
	\Orderlist_{\LL} =
	& \left\{
	\Order \left(
	d_m
	\Rightarrow
	d_{m'}
	\right)
	: m,m' \in \{1,...,M_{\LL}\}, m \ne m', p_{1,m} \in \pos(j_{0,m'}) \text{ or } p_{0,m} \in \pos(j_{1,m'})
	\right\}.
	\end{align*}
	Then
	\begin{align}
	\AncState(\LL)
	&= \Alleles\left(
	\bigcup_{m=1}^{M_{\LL}} D_{\LL, m}
	\right)
	\cap
	\left[
	\bigcap_{\Order(d \Rightarrow d') \in \Orderlist_{\LL}}
    \! \Order(d \Rightarrow d')
	\right]
	\label{eq:leaf_anc_state}
	\end{align}
	where $D_{\LL, m} = \SGWC[p_{0,m}, p_{1,m}, s_{j}, j_{0,m}, j_{1,m}]$.

	Similarly, for any internal node $\N$, we can also write $\AncState(\N)$ in the form of \eqref{eq:leaf_anc_state}.
	If node $\N$ has children nodes $\C_1$, ..., $\C_K$, $\left\{
	D_{\N, m}
	\right\}_{m = 1}^{M_{\N}}$ are pairwise disjoint wildcards and/or singleton-wildcards satisfying \eqref{eq:intersect_sgwc_nonempty} and
	\begin{align*}
	\Orderlist_{\N} &=
	\left\{
	\Order(d \Rightarrow d') \in \left[
	\bigcup_{k = 1}^K
	\Orderlist_{\C_k}
	\right]:
	d \in \left[ \bigcup_{m=1}^{M_{\N}} D_{\N, m}\right]
	\right \}.
	\end{align*}
\end{lemma}
Now that we've shown that $\AncState(\N)$ can be written in terms of disjoint wildcards and singleton-wildcards, we introduce one more notation that will be useful later.
Define $\SG(\N)$ to be the singletons from the singleton-wildcards in $\AncState(\N)$.

\subsection{Likelihood calculation: aggregating states}

Here we show how to use the concept of ``lumpability`` to calculate the approximate likelihood \eqref{eq:approx}, even when there are an infinite number of ancestral states.
Recall that we previously proposed calculating \eqref{eq:approx}, which sums over a subset of all possible ancestral states.
Though this has decreased the set of states to sum over, the number of ancestral states under consideration at each tree node is still infinite.
Even if we ignore the insertion sequences, the number of possible ancestral states still grows at a rate of $O(p^2)$ where $p$ is the number of positions in the unmodified barcode and the likelihood calculation has complexity $O(p^6)$ since we need to construct transition matrices.
Since the barcode is 300 nucleotides long in \citet{McKennaaaf7907}, we cannot calculate the likelihood by considering all possible states separately.

To handle Markov processes with infinite state spaces, one solution is to partition the states into lumped states and show that the behavior of the original Markov process is equivalent to that of an aggregate Markov process over the lumped states \citep{Kemeny1976-ll, Hillston1995-hk}.
This property, called ``lumpability,`` is defined as follows.
\begin{definition}
Let $X(t)$ be a continuous time Markov chain with state space $\Omega$.
If there exists a partition $\{A_1,...,A_M\}$ of $\Omega$ and a continuous time Markov chain $Y(t)$ with state space $\{A_1,...,A_M\}$ such that
\begin{align*}
\Pr(X(t) \in A_i) = \Pr(Y(t) = A_i) \quad \forall i = 1,...,M,
\end{align*}
then $X$ is lumpable.
\end{definition}
\noindent Although lumpability is a well-established technique for Markov chains, the practical difficulty is typically in constructing the appropriate partition \citep{Ganguly2014-dm}.

There is relatively little work on applying these ideas of lumpability in phylogenetics.
(The one application in \citet{Davydov2017-ac} calculates the likelihood of a codon model approximately by assuming states are lumpable, even though this is not necessarily true in their model; here we will show that the states are indeed lumpable.)
Here we extend lumpability to the phylogenetics setting where we have different partitions of the state space at each tree node.
In particular, for some indexing set $B$, define a partition $\{g(b; \N): b\in B\}$ of $\Omega$ at every node $\N$.
Lumpability is only useful for efficient phylogenetic likelihood computation if these partitions are compatible with Felsenstein's pruning algorithm \citep{Felsenstein1981-zs}.
For any $b\in B$ and node $\N$, let $p_{\N}(b)$ be the component of the likelihood for the subtree below $\N$ for states in partition $b$:
\begin{align}
p_{\N}(b) = \Pr\left(
X_{\Leaves(\N)}(T) = a_{\Leaves(\N)}
\middle |
X_{\N}(t_{\N}) \in g(b; \N)
\right )
\label{eq:down_prob}
\end{align}
By Felsenstein's algorithm, we have
\begin{align}
p_{\N}(b)
& = \prod_{\C \in \children(\N)}
\left \{
\sum_{b' \in B}
p_{\C}(b')
\Pr\left (X_{\C}(t_{\C}) \in g(b'; \C) | X_{\C}(0) \in g(b; \N) \right )
\right \}.
\label{eq:felsen_recurs}
\end{align}
For lumpability to be useful, we must be able to show that there exists an easy-to-compute weight function $w(b, \N, b', \C)$ such that
\begin{align}
p_{\N}(b)
& = \prod_{\C \in \children(\N)}
\left \{
\sum_{b' \in B}
p_{\C}(b')
w(b, \N, b', \C)
\Pr\left (X_{\C}(t_{\C}) \in g(b'; \C) | X_{\C}(0) \in g(b; \C) \right )
\right \}
\label{eq:lump_recurs}
\end{align}
and that for each tree node $\C$, $X_{\C}(\cdot)$ is indeed lumpable over the partition $g(\cdot; \C)$.
Obviously if the partitions are the same across all tree nodes, then we can just set all weights to one.
However we will need to construct a different partition for each tree node for the GESTALT likelihood.

We propose partitioning $\Omega$ at tree node $\N$ based on whether or not the allele is a likely ancestral state (i.e. is it in $\AncState(\N)$) and its target status (Figure~\ref{fig:lumpability}):
\begin{definition}
	Define the indexing set $B$ to be $\{0, 1\}^M \cup \{ \other \}$.

	For internal tree node $\N$, partition the state space $\Omega$ into
	\begin{align}
	\begin{cases}
	g(b; \N) = \left\{
	a \in \AncState(\N): \TargStat(a) = b
	\right\}
	& \forall \, b \in \{0, 1\}^M
	\label{eq:grouping}
	\\
	g(\other; \N) = \Omega - \AncState(\N).
	\end{cases}
	\end{align}

	For leaf node $\N$, partition the state space $\Omega$ into
	\begin{align}
	\begin{cases}
	g(b; \N) = \{a_{\N}\} & \text{if } b = \TargStat(a_{\N})\\
	g(b; \N) = \emptyset & \text{if }  b \in \{0, 1\}^M \text{ and } b \ne \TargStat(a_{\N}) \\
	g(\other; \N) = \Omega - \{a_{\N}\}.
	\end{cases}
	\end{align}
\label{def:meta_states}
\end{definition}

To prove that the Markov process over the branch with end node $\N$ is lumpable with respect to the proposed partition, we show that the instantaneous transition rate from any allele in $g(b; \N)$ to the set $g(b'; \N)$ is the same.
Therefore we use $q_{\lump}$ to denote the transition rates between the lumped states $\{g(b; \N)\}$.
The results show that there are two types of transitions between the lumped states, which determines the appropriate formula for calculating $q_{\lump}$.
Either the transition corresponds to an observed indel tract and there is only one indel tract that is a valid for transitioning between the lumped states; or the transition corresponds to a masked indel tract, in which case all indel tracts from the possible target tracts are valid transitions between the lumped states.
\begin{lemma}
	\label{lemma:state_groups}
	Suppose Assumption~\ref{assump:factor_rate} holds.
	Consider any branch with child node $\C$, and target statuses $b, b' \in \{0, 1\}^M$ where the sets $g(b;\C)$ and $g(b';\C)$ are nonempty.
	For any alleles $a, a' \in g(b; \C)$, we have
	\begin{align}
    \begin{split}
	& q_{\lump} \left (g(b; \C), g(b'; \C), t \right )\\
	& =
	\lim_{\Delta \rightarrow 0}
	\frac{
		\Pr\left(
		X_{\C}(t + \Delta) \in g(b'; \C)
		| X_{\C}(t) = a
		\right)
	}{\Delta}\\
	&=
	\lim_{\Delta \rightarrow 0}
	\frac{
		\Pr\left(
		X_{\C}(t + \Delta) \in g(b'; \C)
		| X_{\C}(t) = a'
		\right)
	}{\Delta}.
    \end{split}
	\label{eq:group_the_states}
	\end{align}
	If the only transition from an element in $g(b; \C)$ to $g(b'; \C)$ is via the unique indel $d \in \SG(\C)$ that deactivates the targets $b' \setminus b$, then
	\begin{align*}
	q_{\lump} \left (g(b; \C), g(b'; \C), t \right ) = h(\TT(d), b, t) \Pr(d|\TT(d))
	\end{align*}
	where $h$ is defined in Assumption~\ref{assump:factor_rate}.
	Otherwise, we have
	\begin{align*}
	q_{\lump} \left (g(b; \C), g(b'; \C), t \right ) = \sum_{\tau: \exists d \in \tau = \TT(d) \text{ s.t. } \Apply(d, a) \in g(b'; \N)} h(\tau, b, t).
	\end{align*}
\end{lemma}
\begin{proof}
The instantaneous transition rates for an allele $a \in g(b; \C)$ to the set $g(b'; \C)$ is
\begin{align*}
\lim_{\Delta \rightarrow 0}
\frac{
\Pr\left(
X(t + \Delta) \in g(b'; \C)
| X(t) = a
\right)
}{\Delta}
& =
\sum_{a' \in g(b'; \C)}
q(a, a', t)\\
&=
\sum_{d: \Apply(d, a) \in g(b'; \N)}
h(\tau, b, t)
\Pr(d|\tau).
\end{align*}
If $d$ is an indel tract that can be introduced to the allele $a \in g(b; \C)$ and $\Apply(a, d)$ has target status $b'$, then we can introduce the same indel tract to any other allele $a' \in g(b; \C)$ and $\Apply(a', d)$ will also have the target status $b'$.
Therefore we have proven that \eqref{eq:group_the_states} must hold for all $a, a' \in g(b; \C)$.

To calculate the hazard rate between these lumped states, we rewrite the summation by grouping indel tracts with the same target tract:
\begin{align}
q_{\lump} \left (g(b; \C), g(b'; \C), t \right )
& =
\sum_{\substack{\tau: \exists d \in \tau = \TT(d)\\ \text{ s.t. } \Apply(d, a) \in g(b'; \N)}}
\left\{
\sum_{d \in \tau: \Apply(d, a) \in g(b'; \N)}
h(\tau, b, t)
\Pr(d|\tau)
\right\}.
\label{eq:group_by_TT}
\end{align}
One of the following two cases must be true:
\begin{enumerate}
	\item From the decomposition \eqref{eq:leaf_anc_state} of $\AncState(\C)$, there is only one indel tract $d$ in the sets $D_{\C,m}$ for $m = 1,..., M_{\C}$ such that $\Apply(a, d) \in g(b', \C)$ for all $a\in g(b, \C)$.
	$d$ cannot be from a wildcard or the inner wildcard of a singleton-wildcard since this would contradict the fact that  $d$ is the only indel tract in $\{D_{\C,m}\}$ such that $\Apply(a, d) \in g(b', \C)$ for all $a\in g(b, \C)$.
	Therefore $d$ must be the singleton for some singleton-wildcard $D_{\C,m}$.
	In other words, the only possible transition from $g(b; \C)$ to $g(b'; \C)$ is via the indel tract $d$.
	\item Otherwise, for some target tract $\tau$, there are at least two indel tracts in $d, d' \in \tau$ in the sets $D_{\C,m}$ for $m = 1,..., M_{\C}$ that deactivate targets $b' \setminus b$  such that $\Apply(a, d) \in g(b', \C)$ and $\Apply(a, d') \in g(b', \C)$ for all $a\in g(b, \C)$.
	In this case, $d$ and $d'$ must be from a wildcard or the inner wildcard of a singleton-wildcard ($d$ and $d'$ cannot both be from a singleton of a singleton-wildcard since $d \ne d'$).
	Therefore every indel tract $d$ in $\tau$ satisfies $\Apply(a, d) \in g(b', \C)$ for all $a \in g(b, \C)$.
\end{enumerate}
Therefore \eqref{eq:group_by_TT} simplifies to
\begin{align*}
q_{\lump} \left (g(b; \C), g(b'; \C), t \right )
= \begin{cases}
h(\tau, b, t) \Pr(d|\tau) & \text{if case (1)}\\
\sum_{\tau: \exists d \in \tau = \TT(d) \text{ s.t. } \Apply(d, a) \in g(b'; \N)} h(\tau, b, t) & \text{if case (2)}.
\end{cases}
\end{align*}

Note that to construct the entire instantaneous transition rate matrix of the aggregated process, we can easily calculate the total transition rate away from a target status and then calculate the transition rate to sink state $g(\other; \C)$ using the fact that each row sums to zero.
The transition rate away from $g(\other; \C)$ is zero.
\end{proof}

We are finally ready to combine lumpability with Felsenstein's pruning algorithm.
The following theorem provides a recursive algorithm for calculating \eqref{eq:approx}, using results from above.
\begin{theorem}
	\label{thrm:tree_lik_tt}
	Suppose the above model assumptions hold.
	Consider any tree node $\N$, target status $b$, and nonempty allele group $g(b; \N)$.
	Denote
	\begin{align}
	p_{\N}(b) = \Pr\left(
	X_{\Leaves(\N)}(T) = a_{\Leaves(\N)},
	\left \{
	\{X_{\N'}(t): 0 \le t \le t_{\N'}\} \subseteq \AncState(\N')
	: \N' \in \Desc(\N)
	\right \}
	\middle |
	X_{\N}(t_{\N}) \in g(b; \N)
	\right ).
	\label{eq:approx_down_prob}
	\end{align}
	If $N$ is an internal node, then
	\begin{align}
	p_{\N}(b)
	& = \prod_{\C \in \children(\N)}
	\left \{
	\sum_{
		\substack{
			b' \in \{0,1\}^M\\
			g(b'; \C) \ne \emptyset
		}
	}
	p_{\C}(b')
	\Pr\left (X_{\C}(t_{\C}) \in g(b'; \C) | X_{\C}(0) \in g(b; \C) \right )
	\right \}.
	\label{eq:cond_child}
	\end{align}
where $\Pr\left (X_{\C}(t_{\C}) \in g(b'; \C) | X_{\C}(0) \in g(b; \C) \right )$ is calculated using the instantaneous transition rates given in Lemma~\ref{lemma:state_groups}.
\end{theorem}
\begin{proof}
For any internal node, we know that
\begin{align*}
p_{\N}(b)
& = \prod_{\C \in \children(\N)}
\left \{
\sum_{
	\substack{
		b' \in \{0,1\}^M\\
		g(b'; \C) \ne \emptyset
	}
}
p_{\C}(b')
\Pr\left (X_{\C}(t_{\C}) \in g(b'; \C) | X_{\C}(0) \in g(b; \N) \right )
\right \}.
\end{align*}
(We do not need to sum over the partition $g(\other; \N)$ since it contributes zero probability.)
By irreversibility of the mutation process, $g(b ; \N) \subseteq g(b ; \C)$ if $\C$ is a child of $\N$.
By \eqref{eq:group_the_states} in Lemma~\ref{lemma:state_groups}, $$\Pr\left (\C(t_{\C}) \in g(b'; \C) | \C(0) \in g(b ; \N) \right ) = \Pr\left (\C(t_{\C}) \in g(b'; \C) | \C(0) \in g(b; \C) \right ),$$
which means \eqref{eq:cond_child} also holds for node $\N$.
\end{proof}
\noindent Note that \eqref{eq:cond_child} requires enumerating the possible target statuses.
We can do this quickly using Lemma~\ref{lemma:anc_state_summary}.

Since each node may have a different partition of the state space, we compute a separate instantaneous transition rate matrix for each branch.
If we have multiple barcodes, we need to compute a separate matrix for each branch and for each barcode.
Calculating the likelihood and its gradient can therefore become memory-intensive when there are many branches and/or barcodes.
One way to reduce the amount of memory is to sum over subsets of $\AncState(\cdot)$ instead.
This is often reasonable since it is unlikely for the barcode to have many hidden events.
For the analyses of the zebrafish data, we only sum over states that can be reached when at most one masked indel tract occurs along each branch.
If there are more than 20 such states at a node, we only sum over the possible states when no masked indel tracts occur along that branch.

\subsection{Caterpillar trees}

As discussed in the main manuscript, we resolve the multifurcations in the tree as caterpillar trees to estimate the ordering of events.
Recall that a caterpillar tree is a tree where all the leaves branch off of a central path, which we call the ``caterpillar spine.``
Thus for each multifurcating node $\N$ with children nodes $\C_1,...,\C_K$, we resolve the multifurcation as a tree where the children nodes branch off of the caterpillar spine (Figure~\ref{fig:caterpillar}).
We do not consider all possible resolutions of the multifurcations since there are a super-exponential number of them and we likely do not have enough information to choose between all the possible trees (recall that they are parsimony-equivalent).

We need an efficient method to select the best ordering in each caterpillar tree since the number of possible orderings for $K$ children nodes is $K!$, which is also super-exponential.
Since it is computationally intractable to calculate the likelihood for each ordering, we take an alternate approach where we introduce another approximation of the likelihood.
This approximate likelihood can be computed using the same mathematical expression regardless of the ordering of the children nodes, which means we can tune over all possible orderings in the caterpillar trees by solving a single continuous optimization problem.
\begin{approximation}
We approximate the likelihood by considering only the mutation histories that have a constant allele along the caterpillar spines:
\begin{align}
\Pr \left (
X_{\Leaves(\mathbb{T})}(T) = a_{\Leaves(\mathbb{T})}
\right )
\approx
\Pr \left (
X_{\Leaves(\mathbb{T})}(T) = a_{\Leaves(\mathbb{T})},
\text{alleles are constant on all spines} \right ).
\label{eq:caterpillar_approx}
\end{align}
\label{approx:caterpillar}
\end{approximation}

To construct a mathematical expression for \eqref{eq:caterpillar_approx} that is independent of the ordering along caterpillar trees, we first re-parameterize the branch lengths for children of multifurcating nodes.
For each child $\C$ of a multifurcating node, let $\ell_{\C}$ indicate the distance between the child node and the multifurcating node and $\beta_{\C} \in (0,1)$ indicate the proportion of distance $\ell_{\C}$ that is located on the caterpillar spine (Figure~\ref{fig:caterpillar_math}).
We can capture all possible orderings in a caterpillar tree by varying the values of these two sets of parameters across the children of a multifurcating node.
\begin{figure}
\begin{center}
\includegraphics[width=0.55\textwidth]{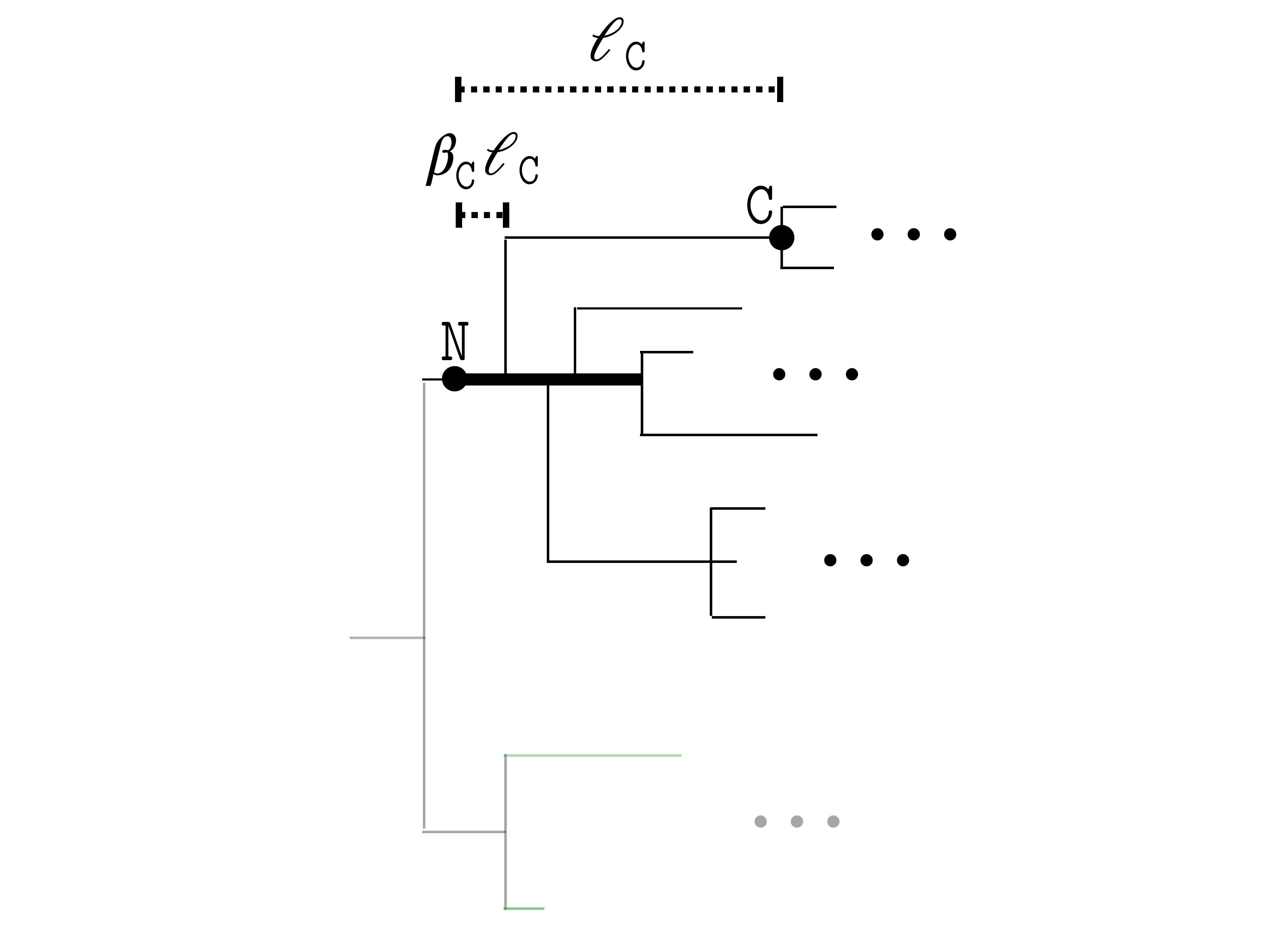}
\end{center}
\caption{
	Parameterization of branch lengths in a caterpillar tree within the context of the entire tree $\mathbb{T}$.
	The rest of the tree is greyed out to draw focus to the caterpillar tree.
	The bolded path is the caterpillar spine.
	Each child node $\C$ of the caterpillar tree is associated with parameters $\ell_{\C}$ and $\beta_{\C}$.
	$\ell_{\C}$ is the length of the path from the start of the caterpillar spine to $\C$.
	$\beta_{\C}$ is the proportion of this path is along the caterpillar spine.
	 The length of the caterpillar spine, $t_{\text{spine}}$, is the maximum value of $\beta_{\C}\ell_{\C}$ over all children nodes $\C$.
}
\label{fig:caterpillar_math}
\end{figure}

Next we extend the likelihood calculation in Theorem~\ref{thrm:tree_lik_tt} where the multifurcations are resolved as caterpillar trees and we want to calculate the approximate likelihood \eqref{eq:caterpillar_approx}.
We do this by recursing on the analogous quantity
\begin{align}
\tilde{p}_{\N}(a) = \Pr \left (X_{\N}(T), \text{alleles are constant on all spines} | \N(0) = a \right ).
\label{eq:caterpillar_prob}
\end{align}
To calculate \eqref{eq:caterpillar_prob}, we use the recursive relation that $\tilde{p}_{\N}(a)$ is equal to
\begin{align}
\begin{cases}
\Pr \left (X_{\N}(t_{\text{spine}}) = a | X_{\N}(0) = a \right)
\prod_{\C \in \children(\N)}
\left \{
\sum_{a' \in \Omega}
\Pr \left (X_{\C}(\ell_{\C} (1 - \beta_{\C}) ) = a' | X_{\C}(0) = a \right)
\tilde{p}_{\C}(a')
\right \}
 &
\text{if } |\children(\N)| > 2\\
\prod_{\C \in \children(\N)}
\left \{
\sum_{
a' \in \Omega
}
\Pr\left (X_{\C}(t_{\C}) = a' | X_{\C}(0) = a \right )
\tilde{p}_{\C}(a')
\right \}
& \ow
\end{cases}
\label{eq:caterpillar}
\end{align}
where $t_{\text{spine}}  = \max\{\ell_{\C} \beta_{\C} : \C \in \children(\N)\}$.
Using the same machinery in Lemma~\ref{lemma:state_groups} and Theorem~\ref{thrm:tree_lik_tt}, we can then calculate \eqref{eq:caterpillar} efficiently by grouping ancestral allele states into lumped states.

\subsection{Penalization}
\label{sec:penalization}
Our algorithm fits the tree and mutation parameters by maximizing the penalized log likelihood.
Penalization improves accuracy when the number of observations is small relative to the number of parameters; GESTALT exhibits this problem because the number of parameters is large and the number of independent barcodes is small (\citet{McKennaaaf7907} only has one barcode).

We propose a tree-based penalty that discourages large differences in the branch lengths $\boldsymbol{\ell}$ and the target rates $\boldsymbol{\lambda}$.
For multifurcating nodes, the branch lengths include the length of the spine as well as the lengths of branches off of the spine, i.e. $\ell_{\C} (1 - \ell'_{\C})$.
Let $L$ be the number of branch lengths in $\boldsymbol{\ell}$.
The penalty is then
\begin{align*}
\Pen_{\boldsymbol{\kappa}}(\theta) =
\kappa_1 \left \|\log\boldsymbol{\lambda} - \frac{1}{M} \sum_{i=1}^M \log(\lambda_i) \right \|_2^2
+ \kappa_2 \left \|\log{\boldsymbol{\ell}} - \frac{1}{L} \sum_{i=1}^L \log(\ell_i) \right \|_2^2
\end{align*}
where  $\kappa_1, \kappa_2 > 0$ are penalty parameters.

Our penalty on the branch length was also considered in the penalized likelihood framework in \citet{Kim2008-rk}.
However the focus of \citet{Kim2008-rk} was to encourage development of methods that were more flexible and had fewer assumptions, rather than to improve estimation in high-dimensional settings.
In particular, their work focused on the standard phylogenetic setting in which there are multiple independent sites.
However, in the GESTALT setting, the current available datasets were generated using a single barcode and we cannot tune the penalty parameters using the common approach of cross-validation \citep{Hastie2009-cb, Arlot2010-rv}.
In addition, \citet{Kim2008-rk} only discussed penalized likelihood in the context of a fixed topology.
In our setting the true tree is unknown and we must consider various tree topologies.

We found that a major hurdle in applying this penalized likelihood framework is that some topologies will naturally have larger penalties.
Therefore we cannot simply choose the tree with the highest penalized log likelihood.
Our solution is to perform a hill-climbing iterative search and score topology updates based on the penalized log likelihood where the penalty is restricted to the shared subtree.
We found that our method tends to choose topology updates that improve the tree estimate (see Figure~\ref{fig:chad_tuning}).

Alternatively, we could have applied regularization methods tailored for the setting where the topology is unknown.
These methods typically regularize the tree towards a prespecified tree \citep{Wu2013-gk, Dinh2018-ww}.
However we would like to place minimal assumptions on the developmental process and we have little to no knowledge about the true tree.

\subsubsection{Tuning penalty parameters}
By varying the value of the penalty parameters $\kappa_1$ and $\kappa_2$, we can control the trade-off between minimizing the penalty versus maximizing the log likelihood.
Choosing appropriate values is crucial for estimation accuracy.
A common approach for tuning penalty parameters is to use cross-validation \citep{Arlot2010-rv}; we use this procedure whenever possible.
Note that we keep the tree topology fixed when tuning the penalty parameter.

We can perform cross-validation when there are multiple barcodes.
First we partition the barcodes into training and valdation sets $T$ and $V$, respectively.
Next we fit tree and mutation parameters $\hat{\ell}_{\boldsymbol{\kappa}}$ and $\hat{\theta}_{\boldsymbol{\kappa}}$, respectively, for each $\boldsymbol{\kappa}$ using only the training data.
We choose the $\boldsymbol{\kappa}$ with the highest validation log likelihood
\begin{align*}
\frac{1}{|V|}\sum_{i \in V} \log \Pr
\left(
X^{(i)}_{\Leaves(\mathbb{T})}(T) = a_{\Leaves(\mathbb{T})}
;
\hat{\ell}_{\boldsymbol{\kappa}}, \hat{\theta}_{\boldsymbol{\kappa}} \right ).
\end{align*}
For our simulation studies with two and four barcodes, we used half of the barcodes for the validation set and half for training.

Unfortunately cross-validation cannot be utilized when there is a single barcode since we cannot split the dataset by barcodes.
Instead we propose a variant of cross-validation described in Algorithm~\ref{algo:cv_for_one}.
The main differences are that we partition the leaves instead of the barcodes into training and validation sets $S$ and $S^c$, respectively; and we select the best penalty parameter that maximizes the conditional probability of the observed alleles at $S^c$ given the observed alleles at $S$.

To partition the leaves, we randomly select a subset of leaf children of each multifurcating node to put in the validation set $S^c$.
We partition leaves in this manner, rather than simply dividing the leaves in half, because we must be able to evaluate (or closely approximate) \eqref{eq:cond_prob_cv} at the end of Algorithm~\ref{algo:cv_for_one} using the fitted branch length and mutation parameters.
That is, we must be able to regraft the leaves in the set $S^c$ onto the fitted tree.
Regrafting is easy for the leaves in our specially-constructed set: The parent node of each leaf in $S^c$ must be located somewhere along the caterpillar spine corresponding to its original multifurcating parent.
In our implementation, we chose to regraft the leaves to the midpoints of their corresponding caterpillar spines.
The regrafting procedure is illustrated in Figure~\ref{fig:cv_one_barcode}.
Note that we do not tune the branch lengths of these validation leaves since it amounts to peeking at the validation data.
In our simulations, we found that when tuning the branch lengths to maximize the unpenalized (or penalized) log likelihood, we nearly almost always choose the smallest penalty parameter since it prioritizes maximizing the likelihood and, therefore, \eqref{eq:cond_prob_cv}.

To assess each candidate penalty parameter $\boldsymbol{\kappa}$, we compare the conditional probability of the observed alleles at $S^c$ given the observed alleles at $S$.
Our motivation is similar to that in cross-validation: If the alleles are observed from the tree with branch and mutation parameters $\ell^*$ and $\theta^*$, we know that
\begin{align}
E\left [ \log \Pr(X_{S^c}(T)|X_{S}(T); \ell^*, \theta^* ); \ell^*, \theta^* \right ]
\ge E \left [\log \Pr(X_{S^c}(T)|X_{S}(T); \ell, \theta ) ; \ell^*, \theta^* \right ]
\qquad  \forall \ell, \theta
\label{eq:exp_cond_prob}
\end{align}
by Jensen's inequality.
(Note that this conditional probability is high only for if we have good estimates of both the mutation parameters and branch lengths of leaves $S^c$. It is not sufficient to only have an accurate estimate of the mutation parameters.)
Recall cross-validation is motivated by a similar inequality but uses $\Pr(X; \ell, \theta )$ rather than a conditional probability.

From a theoretical standpoint, using \eqref{eq:exp_cond_prob} to select penalty parameters makes the most sense if we have an unbiased estimate of the expected conditional probability.
Unfortunately, in our setting, the conditional probability in \eqref{eq:cond_prob_cv} is actually a biased estimate since the fitted parameters depended on the observed alleles at leaves $S$.
Nonetheless, in simulations (where the truth is known), this biased estimate seemed to work well, as the selected penalty parameter was typically close to the best penalty parameter.

\begin{algorithm}
\caption{Cross validation for a single barcode}
\label{algo:cv_for_one}
	\begin{algorithmic}
		\STATE Initialize $S$ to be all the leaves in tree $\mathbb{T}$. Throughout, let $S^c$ denote all the leaves in $\mathbb{T}$ not in $S$.
		\FOR{each multifurcating node $\N$ where at least one children is a leaf}
		\STATE{Let $\C_1,...,\C_m$ be the children nodes of $\N$ that are leaves. Randomly select $m' \ge 1$ of them and remove these from $S$.}
		\ENDFOR
		\STATE{Let $\mathbb{T}_{S}$ be the subtree over the leaves $S$.}
		\FOR{each candidate penalty parameter $\boldsymbol{\kappa}$}
		\STATE{Maximize the penalized log likelihood of the tree $\mathbb{T}_{S}$ with respect to its branch lengths $\boldsymbol{\ell}$ and mutation parameters $\boldsymbol{\theta}$
			$$
			\hat{\boldsymbol{\ell}}_{\boldsymbol{\kappa}}, \hat{\boldsymbol{\theta}}_{\boldsymbol{\kappa}}
			= \argmax_{\boldsymbol{\ell}, \boldsymbol{\theta}}
			\log \Pr \left (X_{S}(T) = a_{S}; \boldsymbol{\ell}, \boldsymbol{\theta} \right)
			+ \Pen_{\boldsymbol{\kappa}}\left( \boldsymbol{\ell}, \boldsymbol{\theta} \right )
			.
			$$
		}
		\ENDFOR
		\STATE{Return the penalty parameter that maximizes the conditional probability:
			\begin{align}
			\hat{\boldsymbol{\kappa}}
			=
			\argmax_{\boldsymbol{\kappa}}
			\Pr \left (X_{S^c}(T) = a_{S^c}
			\mid
			X_{S}(T) = a_{S}; \hat{\boldsymbol{\ell}}_{\boldsymbol{\kappa}}, \hat{\boldsymbol{\theta}}_{\boldsymbol{\kappa}} \right ).
			\label{eq:cond_prob_cv}
			\end{align}
		}
	\end{algorithmic}
\end{algorithm}

\begin{figure}
\begin{center}
\includegraphics[width=0.6\linewidth]{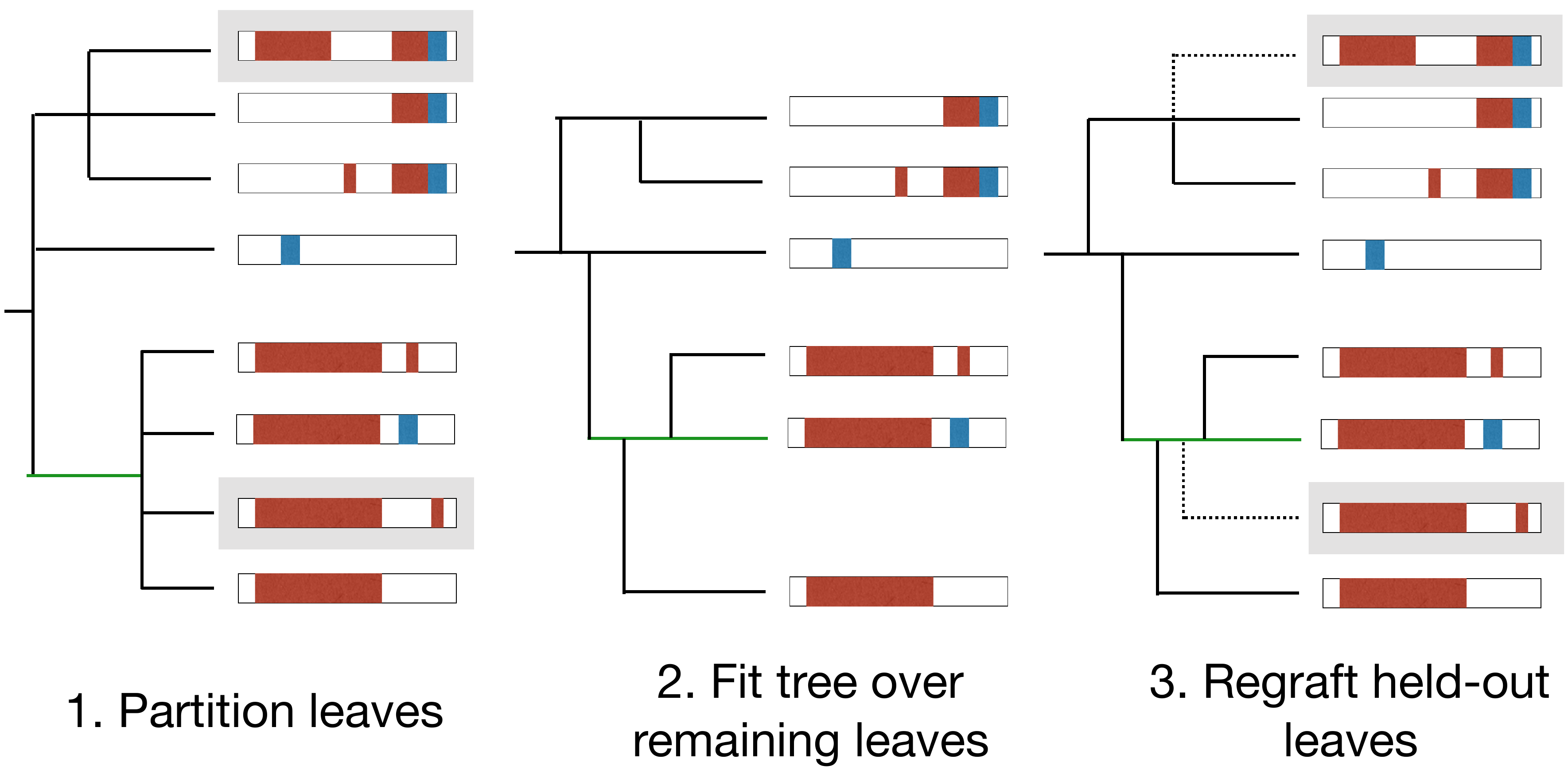}
\end{center}
\caption{
Cross-validation to tune penalty parameters with only one barcode.
We split leaves into training and validation sets $S$ and $S^c$, respectively as follows (left): For each multifurcating node, randomly select a subset of its children that are leaves to put in the ``validation`` set, denoted by the gray boxes.
Fit branch lengths and mutation parameters on the subtree over the remaining leaves (middle).
Regraft the leaves in the ``validation`` set back onto the fitted tree (right).
}
\label{fig:cv_one_barcode}
\end{figure}

Finally, we must simultaneously tune the penalty parameter and the topology of the tree (from Algorithm~\ref{algo:whole_thing}).
Our full algorithm alternates between tuning the penalty parameters for a fixed tree topology and running a single iteration of Algorithm~\ref{algo:whole_thing} for a fixed penalty parameter.
After the penalty parameters are stable, we keep them fixed and only run Algorithm~\ref{algo:whole_thing}.

\subsection{Specific model implementation}

Here we discuss the specific implementation we use to model the GESTALT data.
We suppose the mutation process is homogeneous and does not depend on $t$.
Therefore we will drop the time index $t$ in the model.
Recall that there are two major components of the mutation model: the rates at which target tracts are introduced and the conditional probability of an indel tract given the selected target tract.

To parameterize the rate at which a target tract $\tau = \TT[j_0', j_0, j_1, j_1']$ is introduced, we further decompose the rate into a rate $h_0$ that represents the rate at which the targets $j_0$ and $j_1$ are cut and various scaling factors that control how often deletions are short or long (recall the definition in \eqref{eq:short_long1} and \eqref{eq:short_long2}):
\begin{align*}
h \left (\tau, \TargStat(a) \right )
= h_0 \left (j_0, j_1, \TargStat(a) \right )
\prod_{i = 0}^1
\left[
\gamma_{i} \mathbbm{1} \{j_i \ne j_i' \} + \mathbbm{1} \{j_i = j_i' \}
\right],
\end{align*}
where $\gamma_0$ and $\gamma_1$ parameterize how often long deletions occur to the left and right, respectively.

We specify $h_0$ using the assumption that the cutting time for target $j$ follows an exponential distribution with rate of $\lambda_j > 0$.
For focal target cuts where $j_0 = j_1$, we define
$$
h_0 \left (j_0, j_0, \TargStat(a) \right ) = \lambda_{j_0} \mathbbm{1}\{ \TargStat(j_0, a) = 0 \}.
$$
For double cuts at targets $j_0$ and $j_1$, we suppose the cut time follows an exponential distribution with rate $\omega \cdot (\lambda_{j_0} + \lambda_{j_1})$, where $\omega$ is an additional model parameter that we estimate and does not depend on the targets.

Our parameterization of the double-cut rate is based on the assumption that an inter-target deletion is introduced when the cuts at both targets occur within a small time window of length $\epsilon$.
For random cut times $X_{j_0}$ and $X_{j_1}$ for targets $j_0$ and $j_1$, we approximate that the cut times occur within a small window $\epsilon$ with the distribution
\begin{align}
p \left (|X_{j_0} - X_{j_1} | \le \epsilon, \frac{X_{j_0} + X_{j_1}}{2} = t \right )
& \approx \Pr(|X_{j_0} - X_{j_1}| \le \epsilon) \, p(X_{j_0} = X_{j_1} = t | X_{j_0} = X_{j_1}).
\label{eq:double_cut_approx}
\end{align}
The values on the left and right hand sides approach each other as $\epsilon \rightarrow 0$.
The first component on the right-hand-side of \eqref{eq:double_cut_approx} approaches zero as $\epsilon \rightarrow 0$ and does not vary much for different values of $\lambda_{j_0}, \lambda_{j_1}$ if $\epsilon$ is sufficiently small.
Hence we use the same value of $\omega$ for all targets.
The second component on the right-hand-side of \eqref{eq:double_cut_approx} corresponds to an exponential distribution with the rate $\lambda_{j_0} + \lambda_{j_1}$.

We can interpret $\omega$ in two ways.
First, $\omega$ controls how often a double cut is introduced.
In an unmodified barcode, the relative rate that a double cut is introduced versus a single cut is $\omega \sum_{j_1 < j_2} (\lambda_{j_1} + \lambda_{j_2})$ versus $\sum_{j=1}^M \lambda_{j}$.
The second interpretation, based on \eqref{eq:double_cut_approx}, is that $\omega$ serves as a proxy for $\epsilon$: Larger $\omega$ indicates that an inter-target deletion can be introduced by two cuts spaced farther apart in time.

The second major component of the GESTALT mutation model specifies the conditional probability of introducing a particular indel tract given target tract $\tau = \TT[j_0', j_0, j_1, j_1']$.
An indel tract can be represented by its deletion lengths to the left and right and the insertion sequence.
We will suppose that the probability of a single insertion sequence is uniform over all possible nucleotide sequences of that length.
Let $X_0, X_1, X_2$ be the random variables parameterizing the lengths of the left deletion, right deletion, and insertion, respectively.
Let $x_{\tau, \min, i}$ and $x_{\tau, \max, i}$ for $i = 0$ and $1$ specify the minimum and maximum deletion lengths to the left and right, respectively, for target tract $\tau$.
(For example, if $j_0 = j_0'$, the left deletions must be short so $x_{\tau,\min, i} = 0$ and $x_{\tau,\max, i}$ is the longest deletion without deactivating target $j_0 - 1$.
As another example, if $j_0 = j_0' + 1$, the left deletion is long so $x_{\tau, \min, i}$ is the minimum deletion length to deactivate target $j_0'$ and $x_{\tau, \max, i}$ is the longest length without deactivating target $j_0' - 1$.)
For insertions, $x_{\tau, \min, 2} = 0$ and $x_{\tau, \max, i} = \infty$ regardless of the target tract.

We parameterize the conditional probability of indel $d$ with left deletion, right deletion, and insertion lengths $x_0, x_1, x_2$ given target tract $\tau$ as
\begin{align*}
\Pr\left( X_0 = x_0, X_1 = x_1, X_2 = x_2 | \tau \right)
= p\left(x_0, x_1, x_2 | \tau; x_{\tau, \min, 0}, x_{\tau,\min, 1}, x_{\tau,\min, 2} \right)
\end{align*}
where
\begin{align*}
& p\left( x_0, x_1, x_2 | \tau; x_{\min, 0}, x_{\min, 1}, x_{\min, 2} \right)\\
& =
\begin{cases}
p_{0, \boost}
p\left( x_0,  x_1, x_2 | \tau; x_{\min, 0} + 1, x_{\min, 1}, x_{\min, 2} \right) &\\
+
p_{1, \boost}
p\left( x_0, x_1, x_2 | \tau; x_{\min, 0}, x_{\min, 1} + 1, x_{\min, 2} \right) & \text{ if } x_{\min, 0} = x_{\min, 1} = x_{\min, 2} = 0\\
+
p_{2, \boost}
p\left(x_0, x_1, x_2 | \tau; x_{\min, 0}, x_{\min, 1}, x_{\min, 2} + 1 \right)&\\
& \\
\Pr\left( X_0 = x_0 | \tau; x_{\min, 0} \right)
\Pr\left( X_1 = x_1 | \tau; x_{\min, 1} \right)
\Pr\left( X_2 = x_2 | \tau; x_{\min, 2} \right)
& \ow \\
\end{cases}
\end{align*}
where $p_{0, \boost} + p_{1, \boost} + p_{2, \boost} = 1$.
The probabilities $p_{i, \boost}$ ensure that we can never introduce an indel tract that deletes and inserts nothing.
When the minimum insertion and deletion lengths are zero (in the case of focal target tracts), we use the probabilities $p_{i, \boost}$ to randomly pick whether to boost the minimum left deletion, right deletion, or insertion length by one.

We assume that the deletion lengths follow a zero-inflated, truncated negative binomial distribution; and the insertion lengths follow a zero-inflated negative binomial distribution.
Let $\NBinom(m,q)$ denote the negative binomial distribution, which is the distribution for the number of successes until $m$ failures are observed and $q$ is the probability of success.
The zero-inflation factor for deletion lengths for focal indel tracts is $p_{i, 0}$ and inter-target indel tracts is $p_{i, 1}$, where left and right are indicated by $i = 0$ and $i = 1$, respectively.
The zero-inflation factor for insertion lengths is $p_{2,0} = p_{2,1}$.
Then for $i = 0,1, 2$, we define
\begin{align*}
& \Pr\left( X_i = x_i | \tau; x_{\min, i} \right)\\
& =
\begin{cases}
p_{i, \mathbbm{1}\{j_0 = j_1\}} & \text{if } x_{i} = x_{\min, i} = 0\\
(1 - p_{i, \mathbbm{1}\{j_0 = j_1\}})
\left[
\Pr(X = x - x_{\min, i}; \NBinom(m_i, q_i))
+ \frac{\Pr(X > x_{\max, i}- x_{\min, i}; \NBinom(m_i, q_i))}{x_{\max, i} - x_{\min, i}}
\right]
& \text{if } x_{i} > x_{\min, i}.\\
\end{cases}
\end{align*}

\subsection{Implementation}
The code is implemented in Python using Tensorflow.
We maximize the penalized log likelihood using Adam \citep{Kingma2014-oy}.

\subsection{Comparison Methods}
We use PHYLIP version 3.697 \citep{Felsenstein_undated-wn}, the neighbor-joining algorithm in Bio.Phylo (Biopython version 1.72) \citep{Talevich2012-rt}, and the chronos function in R package \texttt{ape} version 5.2 \citep{ape_package}.

\subsection{Evaluation metrics}
\label{sec:appendix_metrics}
Given ultrametric trees 1 and 2 with the same set of leaves, the \underline{internal node height correlation} between the two trees is calculated as follows (Figure~\ref{fig:internal_node_height}):
\begin{enumerate}
	\item For each internal node in tree 1, find the matching node in tree 2 that is the most recent common ancestor of the same set of leaves.
	\item Calculate the Pearson correlation of the heights of matched nodes.
	\item Do the same swapping tree 1 and 2.
	\item Average the two correlation values.
\end{enumerate}
A correlation of 1 means that the trees are exactly the same; the smaller the correlation is, the less similar the trees are.
\begin{figure}[h!]
\begin{center}
\includegraphics[width=0.8\textwidth]{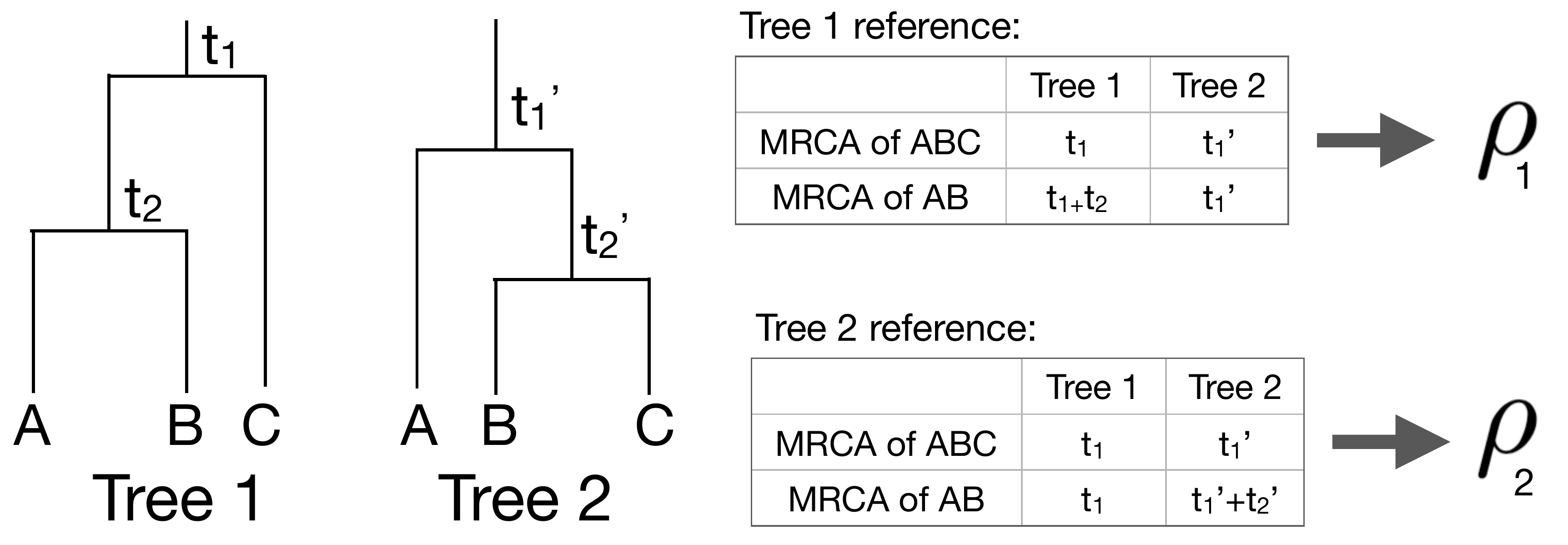}
\end{center}
\caption{
Example calculation of the internal node height correlation.
For each tree, define the groups of leaves based on its internal nodes and calculate the correlation of the time of the most recent common ancestors (MRCAs) of the leaf groups.
The internal node height correlation is the average of the two correlation values.
}
\label{fig:internal_node_height}
\end{figure}

\subsection{Simulation setup and additional results}
\label{sec:supp_simulation_res}

For the results in Figure~\ref{fig:consistency}, the data was simulated with 5 synchronous cell division cycles followed by a birth-death process where the birth rate decayed at a rate of $\exp(-18t)$.
The barcode was composed of six targets with $\boldsymbol{\lambda} = 0.9, 0.85, 0.8, 0.75, 0.7, 0.65, 0.6$.
The weight $\omega$ was set to 0.06 so that 20\% of the unique observed indel tracts were due to double cuts.
We sampled 8\% of the leaves so that the average number of unique observed alleles was around 100 leaves.
We refer to this simulation setup as Simulation A.
We ran 20 replicates of Simulation A.

The results in Figure~\ref{table:big_sim} are from a larger simulation, which we will refer to as Simulation B, that is closer to the data collected in \citet{McKennaaaf7907}.
Since zebrafish data undergo around 11 synchronous cell division cycles, this larger simulation entailed 9 synchronous cell division cycles followed by a birth-death process.
We simulated with a barcode composed of ten targets.
The resulting tree had on average around 200 leaves.
We ran GAPML for 8 topology tuning iterations; at each iteration, we consider at most 15 SPR moves.
The displayed results are from 20 replicates.

For this larger simulation, we also compared the runtimes of the methods on a server with an Intel Xeon 2x8 core processor at 3.20GHz and 256 GB RAM.
Obtaining tree topologies from C-S parsimony and neighbor-joining runs on the order of minutes.
Branch length estimation using \texttt{chronos} runs on the order of seconds.
In contrast, GAPML required up to three hours.
Though the runtime of our method is much longer, it is still reasonable compared to the amount of time spent on data collection, which includes waiting until the organism is a certain age.

Using our simulation engine, we compare two very simple barcode design ideas: a single barcode with many targets, recommended in \citet{Salvador-Martinez2018-dw}, or many identical barcodes.
However we believe the latter is more effective since spreading the targets over separate barcodes tends to create more unique alleles.
In particular, the inter-target deletions tend to be shorter, which means fewer existing mutations are deleted and fewer targets are deactivated.
To test this idea, we compared to the two design options in a simulation setup where we iteratively increased the number of targets by six, i.e. add six targets to the existing barcode or add a new barcode with six targets.
Here we observe all 1024 leaves of a full binary tree with 10 levels.
All targets had the same single-cut rate.
We calibrated the double-cut weight $\omega$ to be around 18\% for both barcode designs -- this slightly favors the single-barcode design since it would have a higher rate of double cuts \textit{in vivo} compared to a multiple-barcode design.
Nevertheless, we find in our simulations that splitting the targets over separate barcodes tends to result in a much larger number of unique alleles than using a single barcode (Figure~\ref{fig:many_vs_one}).
At 30 targets, the multiple-barcode design has roughly 200 more unique alleles on average than the single-barcode design.
Another reason we prefer the multiple-barcode design is that our model and tree estimates improve as the number of independent and identical barcodes increases, as illustrated in Figure~\ref{fig:consistency}.

\begin{figure}
\begin{center}
\includegraphics[width=0.5\linewidth]{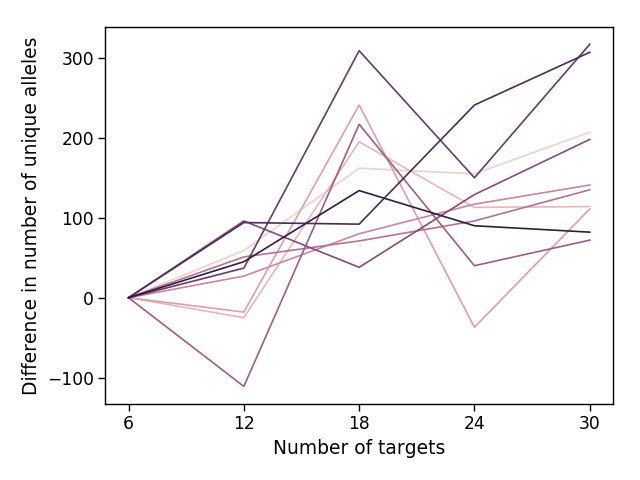}
\end{center}
\caption{
We compare the number of unique alleles obtained GESTALT using a single barcode with many targets versus splitting the targets over multiple independent barcodes.
The alleles are simulated on a full binary tree with 1024 leaves.
Each line corresponds to a simulation where we iteratively add six targets, either by extending the single barcode or adding another barcode with six targets.
A positive difference that the multiple-barcode design has more unique alleles, and vice versa.
}
\label{fig:many_vs_one}
\end{figure}

Next, to better understand our algorithm GAPML, we show in-depth simulation results from a single replicate (Figure~\ref{fig:examples_algo_results}).
Here we use the settings from Simulation B.
Starting from the initial tree topology, the algorithm tunes the branch lengths and mutation parameters to maximize the penalized likelihood.
During the gradient descent algorithm, the BHV distance of the tree estimate decreases  (Figure~\ref{fig:grad_descent}).
In addition, we see that the BHV distance of the tree estimate decreases as Algorithm~\ref{algo:whole_thing} iteratively performs SPR moves to update the tree topology.
\begin{figure}
\begin{subfigure}{0.49\textwidth}
\begin{center}
\includegraphics[width=0.8\textwidth]{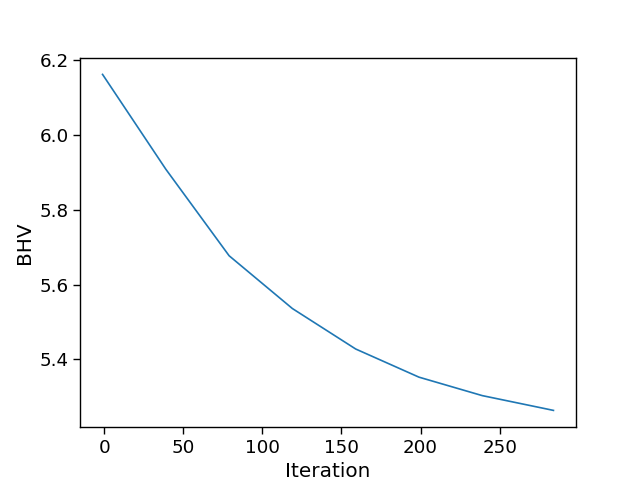}
\end{center}
\caption{
Example of how the BHV distance changes as the branch lengths and mutation parameters are updated using gradient descent to maximize the penalized likelihood.
}
\label{fig:grad_descent}
\end{subfigure}
\begin{subfigure}{0.49\textwidth}
\begin{center}
\includegraphics[width=0.8\textwidth]{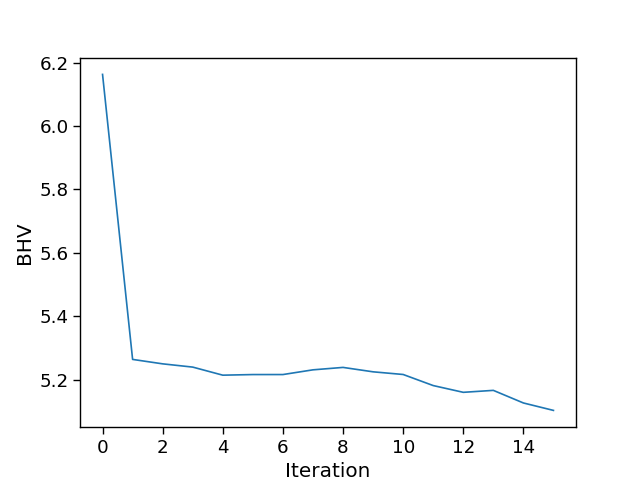}
\end{center}
\caption{
Example of how the BHV distance changes at each SPR iteration, where we select the SPR with the highest likelihood with penalizaton over only the shared tree.
}
\label{fig:chad_tuning}
\end{subfigure}
\caption{
Examples of how the BHV distance changes as the algorithm proceeds for one simulation replicate from the ten-target setting.
}
\label{fig:examples_algo_results}
\end{figure}

Our method searches over the maximally parsimonious trees since they tend to have the highest penalized log likelihood.
To justify this restricted search, we compared the penalized log likelihood for tree topology candidates of different parsimony scores, where the data was generated using Simulation A.
To generate tree topologies with different parsimony scores, we started with the maximally parsimonious tree fit from Camin-Sokal and iteratively applied random SPR moves.
For each of tree rearrangement, we fit a model by maximizing the penalized log likelihood.
The penalty parameter is the same across all rearrangements.
As seen in Figure~\ref{fig:parsimony_log_lik}, the most parsimonious trees have the highest penalized log likelihoods.
Since our method aims to select a tree topology that maximizes the penalized log likelihood, it would not benefit from considering SPR moves that make the tree less parsimonious; instead, considering these additional moves would make the method much slower.
\begin{figure}
\begin{center}
\includegraphics[width=0.5\linewidth]{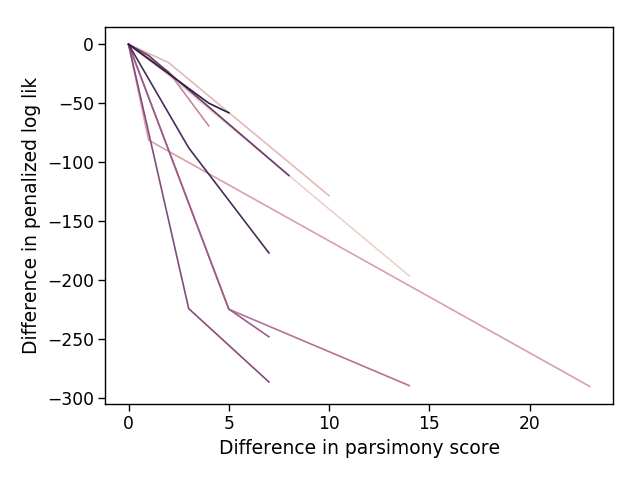}
\end{center}
\caption{
We compare the maximized penalized log likelihood of maximally parsimonious trees to less parsimonious trees.
Each simulation replicate, represented by each line, shows four candidate tree topologies, starting from the most parsimonious one ($x = 0$) to increasingly less parsimonious ones (large differences in parsimony score).
The y-value is the maximized penalized log likelihood of the candidate tree topology minus that of the maximally parsimonious tree.
}
\label{fig:parsimony_log_lik}
\end{figure}

\subsection{Zebrafish data analysis}
For the zebrafish analyses, we estimated the tree over at most 400 randomly selected alleles (without replacement).
50\% of the fish in this dataset had fewer than 400 alleles and the median number of unique alleles over the zebrafish datasets was 443.
25\% of the fish in this dataset had more than 1000 alleles.
We limit the number of alleles analyzed due to runtime restrictions.

To test if the fitted trees are recovering similar developmental relationships across fish rather than random noise, we ran a permutation test comparing the correlation between tissue distances from the estimated trees to that from randomly-estimated trees over randomly-shuffled data.
More specifically, for a given tree topology, we randomly permute the observed alleles at the leaves.
Each allele is associated with the number of times it is observed in each tissue type; we randomly shuffle these abundances over the possible tissue types within each allele.
Finally, we randomly assign branch lengths along the tree by drawing samples from a uniform distribution and using the $t$-parameterization of \citet{Gavryushkin2016-nz} to assign branch lengths.
The correlation between tissue distances in these random trees is close to zero.
All permutation tests were performed using 2000 replicates.

We also tested if the Pearson correlation between the number of tissue types/cell types and the internal node times is different from that of random trees.
The random trees were generated using the same procedure as above.

We conclude by noting that the random trees are generated using the estimated tree topology from each method.
Thus the null distributions are different and the p-values are not directly comparable between methods.
Though this slightly complicates interpretation, we prefer this approach since the estimated tree topology may naturally induce correlation between tissue distances.
For most validation tests, the mean of the null distribution was similar across the different methods, and therefore the p-values are somewhat comparable.
The major exception was the tests that checked recovery of cell-type and germ-layer restriction: here the mean of the null distribution were very different and we abstain from comparing p-values across methods.

For Figure~\ref{table:target_lam_correlation}, we bootstrapped fish replicates to estimate confidence intervals for the average correlation between estimated target cut rates.

\bibliographystyle{plain}
\bibliography{gestalt_arxiv_clean}

\hspace{0.01\textwidth}

\end{document}